\documentclass[11pt]{article}
\usepackage{amsmath, amsthm, amssymb, dsfont, mathrsfs}
\usepackage{fullpage}
\usepackage{times}
\usepackage{float}
\floatstyle{boxed}
\newfloat{constr}{htb!}{lop}
\floatname{constr}{Construction}

\renewcommand{\le}{\leqslant}
\renewcommand{\leq}{\leqslant}

\renewcommand{\geq}{\geqslant}
\title{Capacity of Non-Malleable Codes}

\author{{Mahdi Cheraghchi}\thanks{%
Email: $\langle$mahdi@csail.mit.edu$\rangle$.
Research supported in part by V. Guruswami's Packard Fellowship, MSR-CMU Center for Computational Thinking,
and the Swiss National Science Foundation research grant PA00P2-141980.} \\%
 CSAIL \\
  MIT\\
  Cambridge, MA 02139
 \and %
{Venkatesan Guruswami}\thanks{
Email: $\langle$guruswami@cmu.edu$\rangle$. Research supported in part by the National Science Foundation under Grant No. CCF-0963975.  Any opinions, findings, and conclusions or recommendations expressed in this material are those of the author(s) and do not necessarily reflect the views of the National Science Foundation.
} %
 \\
  Computer Science Department \\
  CMU\\
  Pittsburgh, PA 15213
}

\date{}

\newcommand{\cC}{\mathcal{C}}
\newcommand{\cD}{\mathcal{D}}
\newcommand{\cE}{\mathcal{E}}
\newcommand{\cX}{\mathcal{X}}
\newcommand{\cY}{\mathcal{Y}}

\newcommand{\cS}{\mathcal{S}}
\newcommand{\cN}{\mathcal{N}}

\newcommand{\cU}{\mathcal{U}}
\newcommand{\U}{\mathcal{U}}
\newcommand{\F}{\mathds{F}}
\newcommand{\cF}{\mathcal{F}}

\newcommand{\E}{\mathds{E}}
\newcommand{\supp}{\mathsf{supp}}

\newcommand{\eps}{\epsilon}

\newtheorem{thm}{Theorem}[section]
\newtheorem{prop}[thm]{Proposition}
\newtheorem{claim}[thm]{Claim}
\newtheorem{lem}[thm]{Lemma}

\newtheorem{coro}[thm]{Corollary}
\theoremstyle{definition}

\newtheorem{defn}[thm]{Definition}
\newtheorem{remark}[thm]{Remark}

\newtheorem*{caveat*}{Caveat}

\newcommand{\dist}{\mathsf{dist}}
\newcommand{\distH}{\mathsf{dist}_{h}}

\newcommand{\zo}{\{0,1\}}

\newcommand{\poly}{\mathsf{poly}}

\newcommand{\enc}{{\mathsf{Enc}}}
\newcommand{\dec}{{\mathsf{Dec}}}
\newcommand{\same}{{\underline{\mathsf{same}}}}

\newcommand{\Copy}{\mathsf{copy}}
\newcommand{\distr}{\mathscr{D}}

\newcommand{\mchOK}[1]{}
\newcommand{\vnoteOK}[1]{}

\newcommand{\calF}{\mathcal{F}}

\usepackage[
     backref,
     pagebackref,
     hidelinks,
     pdftitle={Capacity of Non-Malleable Codes},
     pdfsubject={Capacity of Non-Malleable Codes},
     pdfauthor={M. Cheraghchi and V. guruswami}]{hyperref}

\begin{document}

\maketitle

\begin{abstract}
Non-malleable codes, introduced by Dziembowski, Pietrzak and Wichs~(ICS~2010), 
encode messages $s$
in a manner so that tampering the codeword causes the decoder to
either output $s$ or a message that is independent of $s$. While this
is an impossible goal to achieve against unrestricted tampering
functions, rather surprisingly non-malleable coding becomes possible
against every fixed family $\calF$ of tampering functions that is not too large
(for instance, when $|\calF| \le 2^{2^{\alpha n}}$ for some $\alpha
<1$ where $n$ is the number of bits in a codeword).

\smallskip
In this work, we study the ``capacity of non-malleable coding,'' and
establish optimal bounds on the achievable rate as a function of the
family size, answering an open problem from Dziembowski et al.~(ICS~2010). 
Specifically,
\begin{itemize}
\item We prove that for every family $\calF$ with $|\calF| \le
2^{2^{\alpha n}}$, there exist non-malleable codes against $\calF$
with rate arbitrarily close to $1-\alpha$ (this is achieved w.h.p. by
a randomized construction).
\item We show the existence
of families of size $\exp(n^{O(1)} 2^{\alpha n})$ against which there is no
non-malleable code of rate $1-\alpha$ (in fact this is the case w.h.p for a random family of this size). 
\item We also show that $1-\alpha$ is the best achievable rate for the
  family of functions which are only allowed to tamper the first
  $\alpha n$ bits of the codeword, which is of special interest.

As a corollary, this implies that the capacity of
  non-malleable coding in the split-state model (where the tampering
  function acts independently but arbitrarily on the two halves of the
  codeword, a model which has received some attention recently) 
equals $1/2$.
\end{itemize}

We also give an efficient Monte Carlo construction of codes of rate
close to $1$ with polynomial time encoding and decoding that is
non-malleable against any fixed $c > 0$ and family $\calF$ of size
$2^{n^c}$, in particular tampering functions with
say cubic size circuits.

\end{abstract}
\newpage
\tableofcontents
\newpage

\section{Introduction}
Non-malleable codes are a fascinating new concept put forth in
\cite{ref:nmc}, following the program on non-malleable cryptography
which was introduced by the seminal work of Dolev, Dwork and Naor~\cite{ref:DDN00}.
Non-malleable codes are aimed at protecting the integrity of data in situations
where it might be corrupted in ways that precludes error-correction or
even error-detection.  Informally, a code is non-malleable if the
corrupted codeword either encodes the original message, or a completely
unrelated value. This is akin to the notion of non-malleable
encryption in cryptography which requires the intractability of, given
a ciphertext, producing a different ciphertext so that the
corresponding plaintexts are related to each other.

A non-malleable (binary\footnote{Throughout this paper we deal only
  with binary codes. We point out that non-malleability is mainly interesting
  over small alphabets, since when the alphabet size is large enough,
  even error-detection (e.g., via MDS codes) is possible at rates
  achieved by non-malleable codes.
  }) code against a family $\cF$ of tampering functions each mapping $\{0,1\}^n$ to $\{0,1\}^n$, consists of a
randomized encoding function $\enc : \{0,1\}^k \to \{0,1\}^n$ and a
deterministic decoding function $\dec : \{0,1\}^n \to \{0,1\}^k \cup
\{\perp\}$ (where $\perp$ denotes error-detection) which satisfy
$\dec(\enc(s))=s$ always, and the following non-malleability property
with error $\eps$: For every message $s \in \{0,1\}^k$ and every
function $f \in \calF$, the distribution of $\dec(f(\enc(s))$ is
$\eps$-close to a distribution $\cD_f$ that depends only on $f$ and is
independent\footnote{The formal definition (see
  Definition\ref{def:nmCode}) has to accommodate the possibility that
  $\dec$ error-corrects the tampered codeword to the original message
  $s$; and this is handled in a manner independent of $s$ by including
  a special element $\same$ in the support of $\cD_f$.} of $s$.
In other words, if some
  adversary (who has full knowledge of the code and the message $s$,
  but not the internal randomness of the encoder) tampers with the
  codeword $\enc(s)$ corrupting it to $f(\enc(s))$, he cannot control the relationship between $s$
  and the message the corrupted codeword $f(\enc(s))$ encodes.


In general, it is impossible to achieve non-malleability against
arbitrary tampering functions. Indeed, the tampering function can
decode the codeword to compute the original message $s$, flip the last
bit of $s$ to obtain a related message $\tilde{s}$, and then reencode
$\tilde{s}$. This clearly violates non-malleability as the tampered
codeword encodes the message $\tilde{s}$ which is closely related to
$s$. Therefore, in order to construct non-malleable codes, one focuses
on a restricted class of tampering functions. For example, the body of
work on error-correcting codes consists of functions which can flip an
arbitrary subset of bits up to a prescribed limit on the total number
of bit flips. 

The notion of non-malleable coding becomes more interesting for
families against which error-correction is not possible. A simple and
natural such family is the set of functions causing arbitrary
``additive errors," namely $\calF_{\mathsf{add}} = \{ f_\Delta \mid
\Delta \in \{0,1\}^n \}$ where $f_\Delta(x) := x + \Delta$. Note that
there is no restriction on the Hamming weight of $\Delta$ as in the
case of channels causing bounded number of bit flips. While
error-correction is impossible against $\calF_{\mathsf{add}}$,
{\em error-detection} is still possible --- the work of Cramer et
al. \cite{ref:CDFPW08} constructed codes of rate approaching $1$ (which they
called ``Algebraic Manipulation Detection" (AMD) codes) such that
offset by an arbitrary $\Delta \neq 0$ will be detected with high
probability. AMD codes give a construction of non-malleable codes
against the family $\calF_{\mathsf{add}}$.

Even error-detection becomes impossible against many other natural
families of tampering functions. A particularly simple such class
consists of all constant functions $f_{c}(x) := c$ for $c \in
\{0,1\}^n$. This family includes some function that maps all inputs to
a valid codeword $c^\ast$, and hence one cannot detect tampering.
Note, however, that non-malleability is trivial to achieve against
this family --- the rate $1$ code with identity encoding function is
itself non-malleable as the output distribution of a constant function
is trivially independent of the message. A natural function family for
which non-malleability is non-trivial to achieve consists of
{\em bit-tampering functions} $f$ in which the different of bits of the
codewords are tampered independently (i.e., either flipped, set to $0/1$, or 
left unchanged); formally $f(x) = (f_1(x_1),f_2(x_2),\dots,f_n(x_n))$ 
for arbitrary $1$-bit functions $f_1,f_2,\dots,f_n$~\cite{ref:nmc}.

The family $\calF_{\mathsf{all}}$ of all functions $f : \{0,1\}^n \to
\{0,1\}^n$ has size given by $\log \log |\calF_{\mathsf{all}}| = n +
\log n$.  The authors of \cite{ref:nmc} show the existence of a
non-malleable code against {\em any} small enough family $\calF$ (for which
$\log \log |\calF| < n$). The rate of the code is constant if $\log
\log |\calF| \le \alpha n$ for some constant $\alpha \in (0,1)$.  The
question of figuring out the optimal rates of non-malleable codes for
various families of tampering functions was left as an open problem in
\cite{ref:nmc}. In this work we give a satisfactory answer to this question, pinning down the rate for many natural function families. We describe our results next.

\subsection{Our results}
Our results include improvements to the rate achievable as a function of the size of the family of
tampering functions, as well as limitations of non-malleable codes
demonstrating that the achieved rate cannot be improved for natural
families of the stipulated size. Specifically, we establish the
following results concerning the possible rates for non-malleable coding as a function of the size of the family of tampering functions:
\begin{enumerate} 
\item \label{part:exis} (Rate lower bound) We prove in Section~\ref{sec:prob} that if $|\calF|
  \le 2^{2^{\alpha n}}$, then there exists a (strong) non-malleable code of
  rate arbitrarily close to $1-\alpha$ which is non-malleable w.r.t
  $\calF$ with error $\exp(-\Omega(n))$. This significantly improves the 
  probabilistic construction of \cite{ref:nmc}, which achieves a rate close to
  $(1-\alpha)/3$ using a delicate Martingale argument. In particular, for arbitrary
  small families, of size $2^{2^{o(n)}}$, our result shows that the rate can be made
  arbitrarily close to $1$. This was not known to be possible even for
  the family of bit-tampering functions (which has size $4^n$), for
  which $1/3$ was the best known rate\footnote{%
  Assuming the existence of one-way functions, an explicit construction of
  non-malleable codes of rate close to $1$ was proposed in \cite{ref:nmc}.
  This construction, however, only satisfies a weaker definition of non-malleability 
  that considers computational indistinguishability rather than statistical security.%
  }~\cite{ref:nmc}. In fact, we note
  (in Appendix~\ref{sec:barrier}) why the proof strategy of \cite{ref:nmc} is limited to a rate of
  $1/2$ even for a very simple tampering function such as the one that
  flips the first bit.  As discussed in Section~\ref{sec:upper:efficiency}, 
  our probabilistic construction is equipped with an 
  encoder and decoder that can be efficiently and exactly implemented with
  access to a uniformly random permutation oracle and its inverse (corresponding to
  the ideal-cipher model in cryptography). This is a slight additional advantage over
  \cite{ref:nmc}, where only an approximation of the encoder and decoder is shown to be
  efficiently computable.

\item (Upper bound/limitations on rate) The above coding theorem shows
  that the ``capacity'' of a function family $|\calF|$ for
  non-malleable coding is at least $1- (\log \log |\calF|)/n$. We also
  address the natural ``converse coding quesiton" of whether this rate
  bound is the best achievable (Section~\ref{sec:lower}). This turns out to be false in general
  due to the existence of uninteresting large families for which
  non-malleable coding with rate close to $1$ is easy. But we do prove
  that the $1-\alpha$ rate is best achievable in ``virtually all''
  situations:

\begin{enumerate}
\itemsep=0ex
\vspace{-1ex}
\item We prove that for {\em random} families of size $2^{2^{\alpha n}}$, with high probability it is not possible to exceed a rate of $1-\alpha$ for non-malleable coding with small error.
\item \label{part:1-alpha} For the family of tampering functions which leave the last
  $(1-\alpha)n$ bits intact and act arbitrarily on the first $\alpha
  n$ bits, we prove that $1-\alpha$ is the best achievable rate for
  non-malleable coding. (Note that a rate of $1-\alpha$ is trivial to
  achieve for this family, by placing the message bits in the last
  $(1-\alpha)n$ bits of the codeword, and setting the first $\alpha n$
  bits of the codeword to all $0$s.)
\end{enumerate}

\end{enumerate}

The result \ref{part:1-alpha}, together with the existential result
\ref{part:exis} above, pins down the optimal rate for non-malleable
codes in the {\em split-state model} to $1/2$. In the split-state
model, which was the focus of a couple of recent
works~\cite{ref:DKO,ref:ADL}, the tampering function operates
independently (but in otherwise arbitary ways) on the two halves of
the codeword, i.e., $f(x) = ((f_1(x_1),f_2(x_2))$ where $x_1,x_2$ are
the two halves of $x$ and $f_1,f_2$ are functions mapping $n/2$ bits
to $n/2$ bits. The recent work \cite{ref:ADL} gave an explicit
construction in this model with polynomially small rate. Our work
shows that the capacity of the split-state model is $1/2$, but we do
not offer any explicit construction. For the more restrictive class of
bit-tampering functions (where each bit is tampered independently), in
a follow-up work~\cite{ref:CG2} we give an explicit construction with
rate approaching $1$~\cite{ref:CG2}. We also present in that work a
reduction of non-malleable coding for the split-state model to a new
notion of non-malleable two-source extraction.

\medskip \noindent {\bf Monte Carlo construction for small
  families.} Our result \ref{part:exis} above is based on a random
construction which takes exponential time (and space). Derandomizing
this construction, in Section~\ref{sec:MC} we are able to obtain an efficient Monte Carlo
construction of non-malleable codes of rate close to $1$ (with
polynomial time encoding and decoding, and inverse polynomially small
error) for an {\em arbitrary} family of size
$\exp(n^c)$ for any fixed $c > 0$. Note that in particular this includes
tampering functions that can be implemented by circuits of any fixed polynomial size,
or simpler families such as bit-tampering adversaries.
The construction does not rely on any computational hardness assumptions,
at the cost of using a small amount of randomness.

\subsection{Proof ideas}
  
\smallskip \noindent {\bf Rate lower bound.} Our construction
of rate $\approx 1-(\log \log |\cF|)/n$ codes is obtained by picking for each
message, a random blob of $t$ codewords, such that blobs corresponding
to distinct messages are disjoint.  For each tampering function $f$, 
our proof analyzes the distribution of $\dec(f(\enc(s))$ for each message $s$ separately, 
and shows that w.h.p. they are essentially close to the same distribution $\cD_f$.
In order to achieve sufficiently small error probability allowing for a union
bound, the proof uses a number of additional ideas, including a
randomized process that gradually reveals information about the code
while examining the $t$ codewords in each blob in sequence. The
analysis ensures that as little information is revealed in each step
as possible, so that enough independence remains in the
conditional joint distribution of the codewords throughout the analysis.
Finally, strong concentration bounds are used to derive the desired
bound on the failure probability.
The proof for the special case of bijective tampering functions turns
out to be quite straightforward, and as a warm-up we present this special
case first in Section~\ref{sec:upper:bijection}.

\medskip \noindent {\bf Monte Carlo construction.} 
Since the analysis of the probabilistic code construction considers
each message $s$ separately, we observe that it only
only needs limited ($t$-wise) independence of the
codewords. On the other hand, the code construction is designed to be
sparse, namely taking $t = \mathrm{poly}(n,\log |\cF|,1/\eps)$ suffices
for the analysis. This is the key idea behind our efficient Monte Carlo
construction for small families with $\log |\cF| \le
\mathrm{poly}(n)$.

 The birthday paradox implies that picking the blob of codewords
 encoding each message independently of other messages, while
 maintaining disjointness of the various blobs, limits the rate to $1/2$. Therefore, we construct the code by means of a $t$-wise independent
 {\em decoding} function implemented via a random low-degree
 polynomial. After overcoming some complications to ensure an
 efficient encoding function, we get our efficient randomized
 construction for small families of tampring functions.


\medskip \noindent {\bf Rate upper bounds.} 
Our main impossibility result for the family of adversaries that
only tamper the first $\alpha n$ bits of the codeword uses
an information theoretic argument. We argue that if the rate
of the code is sufficiently large, one can always find messages
$s_0$ and $s_1$ and a set $X_\eta \subseteq \zo^{\alpha n}$ such that
the following holds: The first $\alpha n$ bits of the encoding of $s_0$ has a 
noticeable chance of being in $X_\eta$, whereas this chance for
$s_1$ is quite small. Using this property, we design an adversary
that maps the first $\alpha n$ bits of the encoding to a dummy
string if they belong to $X_\eta$ and leaves the codeword intact otherwise.
This suffices to violate non-malleability of the code.

\section{Preliminaries} \label{sec:prelim}
\subsection{Notation}
We use $\U_n$ for the uniform distribution on $\zo^n$ and $U_n$
for the random variable sampled from $\U_n$ and independently 
of any existing randomness.
For a random variable $X$, we denote by $\distr(X)$ the probability distribution
that $X$ is sampled from. Moreover, for an event $\cE$, we use $\distr(X | \cE)$ to denote
the conditional distribution on the random variable $X$ on the event $\cE$.
Generally, we will use calligraphic symbols (such as $\cX$) for probability distributions
and the corresponding capital letters (such as $X$) for related random variables.
For a discrete distribution $\cX$, we denote by $\cX(x)$ the probability mass assigned to $x$
by $\cX$.
Two distributions $\cX$ and $\cY$ being $\eps$-close in statistical distance is denoted by
$\cX \approx_\eps \cY$. We will use $(\cX, \cY)$ for the product distribution
with the two coordinates independently sampled from $\cX$ and $\cY$.
All unsubscripted logarithms are taken to the base $2$.
Support of a discrete random variable (or distribution) $X$ is denoted by $\supp(X)$.
With a slight abuse of notation, for various bounds we condition probabilities
and expectations on random variables rather than events (e.g., $\E[X|Y]$, or
$\Pr[\cE|Y]$). In such instances, the notation means that the statement holds
for \emph{every} possible realization of the random variables that we condition on.


\subsection{Definitions}

In this section, we review the formal definition of non-malleable codes as introduced
in \cite{ref:nmc}. First, we recall the notion of \emph{coding schemes}.

\begin{defn}[Coding schemes] \label{def:scheme}
A pair of functions $\enc\colon \zo^k \to \zo^n$ and $\dec\colon \zo^n \to \zo^k \cup \{\perp\}$
where $k \leq n$ 
is said to be a coding scheme with block length $n$ and message length $k$
if the following conditions hold.
\begin{enumerate}
\itemsep=0ex
\vspace{-1ex}
\item The encoder $\enc$ is a randomized function; i.e., at each call it receives a 
uniformly random sequence of coin flips that the output may depend on. This random input
is usually omitted from the notation and taken to be implicit. Thus for any 
$s \in \zo^k$, $\enc(s)$ is a random variable over $\zo^n$. The decoder $\dec$ is;
however, deterministic.

\item For every $s \in \zo^k$, we have $\dec(\enc(s)) = s$ with probability $1$.
\end{enumerate}

The \emph{rate} of the coding scheme is the ratio $k/n$.
A coding scheme is said to have relative distance $\delta$, for some $\delta \in [0,1)$,
if for every $s \in \zo^k$ the following holds. Let $X := \enc(s)$. Then,
for any $\Delta \in \zo^n$ of Hamming weight at most $\delta n$,
$\dec(X + \Delta) = \perp$ with probability $1$. \qed
\end{defn}

\noindent Before defining non-malleable coding schemes, we find it convenient to define the following notation.

\begin{defn}
For a finite set $\Gamma$, the function $\Copy\colon (\Gamma \cup \{\same\}) \times \Gamma \to 
\Gamma$ is defined as follows:
\[
\Copy(x, y) := \begin{cases}
x & x \neq \same, \\
y & x = \same.
\end{cases} \qquad\qquad\qquad\qquad\qed
\]
\end{defn}

\noindent
The notion of non-malleable coding schemes from \cite{ref:nmc}
can now be rephrased as follows.

\begin{defn}[Non-malleability] \label{def:nmCode}
A coding scheme $(\enc, \dec)$ with message length $k$ and block length $n$
is said to be non-malleable with error $\eps$ (also called \emph{exact security})
with respect to a family
$\cF$ of tampering functions acting on $\zo^n$ (i.e., each $f \in \cF$ maps
$\zo^n$ to $\zo^n$) if for every $f \in \cF$ there is a distribution
$\cD_f$ over $\zo^k \cup \{\perp, \same\}$ such that the following holds.
Let $s \in \zo^k$ and define the random variable $S := \dec(f(\enc(s)))$. 
Let $S'$ be independently sampled from $\cD_f$. Then,
\[
\distr(S) \approx_\eps \distr(\Copy(S', s)). \qquad\qquad\qed
\]
\end{defn}

\begin{remark} \label{rem:perp}
The above definition allows the decoder to output a special symbol
$\perp$ that corresponds to error detection. It is easy to note that
any such code can be transformed to one where the decoder never outputs
$\perp$ without affecting the parameters (e.g., the new decoder may simply
output $0^k$ whenever the original decoder outputs $\perp$).
\end{remark}

Dziembowski et al.~\cite{ref:nmc} also consider the following stronger
variation of non-malleable codes.

\begin{defn}[Strong non-malleability] \label{def:nmCode:strong}
A pair of functions as in Definition~\ref{def:nmCode} is 
said to be a \emph{strong} non-malleable coding scheme with error $\eps$ with respect to a family
$\cF$ of tampering functions acting on $\zo^n$ if the conditions 
$(1)$ and $(2)$ of Definition~\ref{def:nmCode} is satisfied, and additionally,
the following holds. For any message $s \in \zo^k$,
let $E_s := \enc(s)$, consider the random variable
\[
D_s := \begin{cases}
\same & \text{if $f(E_s) = E_s$,} \\
\dec(f(E_s)) & \text{otherwise,}
\end{cases}
\] and let $\cD_{f,s} := \distr(D_s)$.
It must be the case that for every pair of distinct messages $s_1, s_2 \in \zo^k$,
$\cD_{f,s_1} \approx_\eps \cD_{f, s_2}$. \qed
\end{defn}

\begin{remark}[Computational security] \label{rem:computational}
Dziembowski et al. also consider the case where statistical distance is replaced
with computational indistinguishability with respect to a bounded computational model.
As our goal is to understand information-theoretic limitations of
non-malleable codes, we do not consider this variation in this work.
It is clear, however, that our negative results in Section~\ref{sec:lower}
apply to this model as well.
A related (but incomparable) model that we consider in Section~\ref{sec:MC} is when the
distinguishability criterion is still statistical; however the adversary
is computationally bounded (e.g., one may consider the family of polynomial sized Boolean
circuits). For this case, we construct an efficient Monte Carlo coding scheme 
that achieves any rate arbitrarily close to $1$.
\end{remark}

\begin{remark}[Efficiency of sampling $\cD_f$] \label{rem:efficiency}
The original definition of non-malleable codes in \cite{ref:nmc} also requires
the distribution $\cD_f$ to be efficiently samplable given oracle access
to the tampering function $f$. We find it more natural to remove this 
requirement from the definition since even combinatorial non-malleable codes
that are not necessarily equipped with efficient components (such as
the encoder, decoder, and sampler for $\cD_f$) are interesting and
highly non-trivial to construct. It should be noted; however, that
for any non-malleable coding scheme equipped with an efficient encoder
and decoder, the following is a valid and efficiently samplable 
choice for the distribution $\cD_f$
(possibly incurring a constant factor increase in the error parameter):  
\begin{enumerate}
\itemsep=0ex
\vspace{-1ex}
\item Let $S \sim \U_k$, and $X := \enc(S)$.
\item If $\dec(X) = S$, output $\same$. Otherwise, output $\dec(X)$.
\end{enumerate}
Our Monte Carlo construction in Section~\ref{sec:MC} is 
equipped with a polynomial-time encoder and decoder.
So is the case for our probabilistic construction in 
Section~\ref{sec:prob} in the random oracle model.
\end{remark}

\section{Probabilistic construction of non-malleable codes}
\label{sec:prob}

In this section, we introduce our probabilistic construction of non-malleable codes.
Contrary to the original construction of Dziembowski et al.~\cite{ref:nmc}, 
where they pick a uniformly random truth table for the decoder and 
do not allow the $\perp$ symbol. Our code, on the other hand, is quite sparse.
In fact, in our construction $\dec(U_n) = \perp$ with high probability.
As we observe in Section~\ref{sec:barrier}, this is the key to our improvement,
since uniformly random decoders cannot achieve non-malleability even against
extremely simple adversaries at rates better than $1/2$. Moreover,
our sparse construction offers the added feature of having a large minimum
distance in the standard coding sense; any tampering scheme that perturbs
the codeword in a fraction of the positions bounded by a prescribed limit
will be detected by the decoder with probability $1$.
Another advantage of sparsity is allowing a compact representation for
the code. We exploit this feature in our Monte Carlo construction of
Section~\ref{sec:MC}.
Our probabilistic coding scheme is described in Construction~\ref{constr:prob}.

We remark that Construction~\ref{constr:prob} can be efficiently implemented
in the ideal-cipher model, which in turn implies an efficient approximate
implementation in the random oracle model (see the discussion 
following the proof of Theorem~\ref{thm:upperBound}
in Section~\ref{sec:upper:efficiency}). In turn, this implies that the
distribution $\cD_f$ in Definition~\ref{def:nmCode} for this construction
can be efficiently sampled in both models (see Remark~\ref{rem:efficiency}).

\begin{constr} 
  \caption{Probabilistic construction of non-malleable codes.}

  \begin{itemize}
  \item {\it Given: } Integer parameters $0 < k \leq n$ and integer $t > 0$
  such that $t 2^k \leq 2^n$, and a relative distance parameter $\delta$, $0 \le \delta < 1/2$.

  \item {\it Output: } A pair of functions $\enc\colon \zo^k \times \zo^n$
  and $\dec\colon \zo^n \to \zo^k$, where $\enc$ may also use a uniformly random
  seed which is hidden from that notation, but $\dec$ is deterministic.

  \item {\it Construction: } 
   \begin{enumerate}
   \itemsep=0ex
  \item Let $\cN := \zo^n$. 
  \item For each $s \in \zo^k$, in an arbitrary order, 
  \begin{itemize}
  \item Let $E(s) := \emptyset$.
  \item For $i \in \{1, \ldots, t\}$:
  \begin{enumerate}
  \item Pick a uniformly random vector $w \in \cN$.
  \item Add $w$ to $E(s)$.
  \item Let $\Gamma(w)$ be the Hamming ball of radius $\delta n$ centered at $w$.
  Remove $\Gamma(w)$ from $\cN$ (note that when $\delta = 0$, we have $\Gamma(w) = \{w\}$).
  \end{enumerate}
  \end{itemize}
  \item Given $s \in \zo^k$, $\enc(s)$ outputs an element of $E(s)$ uniformly
  at random.
  
  \item Given $w \in \zo^n$, $\dec(s)$ outputs the unique $s$ such that
  $w \in E(s)$, or $\perp$ if no such $s$ exists.
\end{enumerate}     
  \end{itemize}
  \label{constr:prob}
\end{constr}

The main theorem of this section is the result below that proves
non-malleability of the coding scheme in Construction~\ref{constr:prob}.

\begin{thm} \label{thm:upperBound}
Let $\cF\colon \zo^n \to \zo^n$ be any family of tampering functions.
For any $\eps, \eta > 0$, with probability at least $1-\eta$,
the coding scheme $(\enc, \dec)$ of Construction~\ref{constr:prob}
is a strong non-malleable code with respect to $\cF$ and with error $\eps$
and relative distance $\delta$,
provided that both of the following conditions are satisfied.
\begin{enumerate}
\item $t \geq t_0$, for some 
\begin{equation} \label{eqn:thm:upper:t}
t_0 = O\left( \frac{1}{\eps^6} \Big(\log\frac{|\cF|N}{\eta} \Big) \right).
\end{equation}
\item $k \leq k_0$, for some 
\begin{equation} \label{eqn:thm:upper:k0}
k_0 \geq n(1-h(\delta))-\log t-3\log(1/\eps)-O(1),
\end{equation}
where $h(\cdot)$ denotes the binary entropy function.
\end{enumerate}
Thus by choosing $t = t_0$ and $k=k_0$, the construction satisfies
\[
k \geq n(1-h(\delta)) - \log\log(|\cF|/\eta) - \log n - 9 \log(1/\eps) - O(1).
\]
In particular, if $|\cF| \leq 2^{2^{\alpha n}}$ for any constant $\alpha \in (0, 1)$,
the rate of the code can be made arbitrarily close to $1-h(\delta)-\alpha$ while allowing
$\eps = 2^{-\Omega(n)}$.
\end{thm}

\begin{remark}(Error detection)
An added feature of our sparse coding scheme is the error-detection capability.
However, observe that any probabilistic coding scheme that is non-malleable against
all families of adversaries of bounded size over $\zo^n$ (such as Construction~\ref{constr:prob},
Construction~\ref{constr:monte}, and the probabilistic construction of \cite{ref:nmc}) 
can be turned into one having relative distance $\delta$ (and satisfying the same
non-malleability guarantees) by composing the construction
with a fixed outer code $\cC$ of block length $n$ and relative distance $\delta$. Indeed, 
any class $\cF$ of tampering functions for the composed code corresponds to a class $\cF'$
of the same size or less for the original construction. Namely, each function $f' \in \cF'$ equals
$\dec_\cC \circ f$ ($\dec_\cC$ being the decoder of $\cC$) 
for some $f \in \cF$. The caveat with this approach (rather than
directly addressing distance as in Construction~\ref{constr:prob}) is that the composition
may lose strong non-malleability even if the original code is strongly
non-malleable. Indeed, it may be the case that $f$ is 
a sophisticated tampering function whereas its projection $f'$ becomes as simple as
the identity function. If so, non-malleability may be satisfied by choosing
$\cD_f := \distr(\same)$ whereas strong non-malleability does not hold.
\end{remark}

\subsection{Proof of Theorem~\ref{thm:upperBound} for bijective adversaries}
\label{sec:upper:bijection}

We first prove the theorem for adversaries that are bijective and have
no fixed points. This case is still broad enough to contain interesting
families of adversaries such as additive error adversaries $\cF_{\mathsf{add}}$ mentioned in the introduction, for which case we reconstruct the
existence proof of AMD codes (although optimal explicit
constructions of AMD codes are already known \cite{ref:CDFPW08}).

As it turns out, the analysis for this case is quite straightforward,
and significantly simpler than the general case that we will address
in Section~\ref{sec:upper:general}.

Let $N := 2^n$, $K := 2^k$, and consider a fixed message $s \in \zo^k$
and a fixed bijective tampering function $f\colon \zo^n \to \zo^n$
such that for all $x \in \zo^n$, $f(x) \neq x$. We show that the non-malleability
requirement of Definition~\ref{def:nmCode} holds with respect to the
distribution $\cD_f$ that is entirely supported on $\{ \perp \}$.
That is, we wish to show that with high probability, the coding
scheme $(\enc, \dec)$ of Construction~\ref{constr:prob} is so that
\begin{equation}
\label{eqn:upperSimple:perp}
\Pr[\dec(f(\enc(s))) \neq \perp] \leq \eps.
\end{equation}
By taking a union bound over all choices of $f$ and $s$, this would imply
that with high probability, the code is non-malleable (in fact, strongly 
non-malleable) for the entire family $\cF$.

Let $E(s) := \supp(\enc(s))$ be the set of the $t$ codewords that are mapped to $s$
by the decoder. Let $E_1, \ldots, E_t$ be the codewords in this set in the order
they are picked by the code construction. For any $x \in \zo^n \setminus E(s)$,
we know that
\begin{equation*} 
\Pr[\dec(x) \neq \perp] \leq t(K-1)/(N-t) \leq \frac{\gamma}{1-\gamma} \ ,
\end{equation*}
where $\gamma := tK/N$.
This can be seen by observing that each codeword in $E(s')$ for $s' \neq s$
is uniformly distributed on the set $\zo^n \setminus E(s)$, and taking a union
bound.
Thus, in particular since $\{ f(E_1), \ldots, f(E_t) \}$ is a set of size $t$
outside $E(s)$, we see that 
$\Pr[\dec(f(E_1)) \neq \perp] \leq \frac{\gamma}{1-\gamma}$.
In fact, the same argument holds for $\dec(E_2)$ conditioned on any
realization of $f(E_1)$, and more generally, one can derive for each $i \in [t]$,
\begin{equation} \label{eqn:upperSimple:perpB}
\Pr[\dec(f(E_i)) \neq \perp | f(E_1), \ldots, f(E_{i-1})] \leq \frac{\gamma}{1-\gamma}.
\end{equation}
Define indicator random variables $0=X_0, X_1, \ldots, X_t \in \zo$, where
$X_i = 1$ iff $\dec(f(E_1)) \neq \perp$. From \eqref{eqn:upperSimple:perpB}
and using Proposition~\ref{prop:restriction}, we can deduce that
for all $i \in [t]$,
$\Pr[X_i = 1 | X_0, \ldots, X_{i-1}] \leq \tfrac{\gamma}{1-\gamma}$.
Now, using Proposition~\ref{prop:chernoff:dependent}, letting
$X := X_1 + \cdots + X_t$, 
\[
\Pr[X > \eps t] \leq \Big( \frac{e \gamma}{\eps(1-\gamma)} \Big)^{\eps t}.
\]
Assuming $\gamma \leq \eps/4$, the above upper bound simplifies to
$\exp(-\Omega(\eps t))$. By taking a union bound over all possible
choices of $s$ and $f$ (that we trivially upper bound by $N |\cF|$),
it can be seen that, as long as $t \geq t_0$ for some choice of
$t_0 = O\Big( \tfrac{1}{\eps} \log(\tfrac{N|\cF|}{\eta})  \Big)$,
the probability that $(\enc, \dec)$ fails to satisfy 
\eqref{eqn:upperSimple:perp} for some choice of $s$ and $f$ is
at most $\eta$.

Finally, observe that the assumption $\gamma \leq \eps/4$ can
be satisfied provided that $K \leq K_0$ for some choice of
$K_0 = \Omega(\eps N/t)$, 
or equivalently, when $k \leq k_0$ for some choice of
$k \geq n - \log t - \log(1/\eps)$.
Note that for this case the proof obtains a better dependence on $\eps$
compared to \eqref{eqn:thm:upper:t} and \eqref{eqn:thm:upper:k0}.

\subsection{Proof of Theorem~\ref{thm:upperBound} for general adversaries}
\label{sec:upper:general}
\newcommand{\reveal}{\mathsf{Reveal}}

First, we present a proof sketch describing the ideas an intuitions behind
the general proof, and then proceed with a full proof of the theorem.

\subsection*{$\bullet\ $ Proof sketch}

%
In the proof for bijective adversaries, we heavily used the fact that
the tampering of each set $E(s)$ of codewords is a disjoint set of the same size.
For general adversaries; however, this may not be true. Intuitively,
since the codewords in $E(s)$ are chosen uniformly and almost independently
at random (ignoring the distinctness dependencies), the tampered distribution
$f(E(s))$ should look similar to $f(\U_n)$ for all $s$, if $|E(s)|$ is sufficiently large.
Indeed, this is what shown in the proof. The proof also adjusts the 
probability mass of $\same$ according to the fraction of the fixed
points of $f$, but we ignore this technicality for the proof sketch.

Note that the distribution $f(\U_n)$ may be arbitrary, and may assign
a large probability mass to a small set of the probability space.
For example, $f$ may assign half of the probability mass to a single
point.
We call the points in $\zo^n$ such that receive a noticeable share
of the probability mass in $f(\U_n)$ 
the \emph{heavy elements} of $\zo^n$, and fix the randomness
of the code construction so that the decoder's values at heavy elements
are revealed before analyzing each individual message $s$.
Doing so allows us to analyze each message $s$ separately
and take a union bound on various choices of $s$ as in the case of
bijective adversaries. Contrary to the bijective case; however,
the distribution $\cD_f$ is no longer entirely supported on $\perp$;
but we show that it still can be made to have a fairly small support;
roughly $\poly(n, \log|\cF|)$. 
More precisely, the proof shows non-malleability with respect to
the choice of $\cD_f$ which is explicitly defined to be the distribution
of the following random variable:
\begin{equation*} 
D := \begin{cases}
\same & \text{if $f(U_n) = U_n$}, \\
\dec(f(U_n)) & \text{if $f(U_n) \neq U_n$ and $f(U_n) \in H$}, \\
\perp & \text{otherwise,}
\end{cases}
\end{equation*}
where $H \subseteq \zo^n$ is the set of heavy elements formally
defined as
\[
H := \{ x \in \zo^n\colon \Pr[f(U_n) = x] > 1/r \},
\]
for an appropriately chosen $r = \Theta(\eps^2 t)$.

Although the above intuition is natural, 
turning it into a rigorous proof requires substantially more work than
the bijective case, and the final proof turns out to be
rather delicate even though it only uses elementary 
probability tools. The first subtlety is that revealing the
decoder at the heavy elements creates dependencies between
various random variables used in the analysis. In order to
make the proof more intuitive, we introduce a random process,
described as an algorithm $\reveal$, that gradually reveals information about
the code as the proof considers the codewords $E_1, \ldots, E_t$
corresponding to the picked message $s$. The process outputs a
list of elements in $\zo^k$, and we show that the empirical distribution 
of this list is close to the desired $\cD_f$ for all messages $s$. 

Roughly speaking, at each step $i \in [t]$ the analysis estimates the distribution of 
$\dec(f(E_i))$ conditioned on the particular realizations 
of the previous codewords. There are three subtleties that we need
to handle to make this work:
\begin{enumerate}
\item The randomness corresponding to some of the $E_i$ is previously revealed 
by the analysis and thus such codewords cannot be assumed to be uniformly
distributed any more. This issue may arise due to the revealing of the
decoder's values at heavy elements in the beginning of analysis, or
existence of cycles in the evaluation graph of the tampering function $f$.
Fortunately, it is straightforward to show that the number of such codewords
remain much smaller than $t$ with high probability, and thus they may simply be ignored.

\item At each step of the analysis, the revealed information make
the distribution of $\dec(f(E_i))$ gradually farther from the desired
$\cD_f$. The proof ensures that the expected increase at each step
is small, and using standard Martingale concentration bounds the total deviation
from $\cD_f$ remains sufficiently small with high probability at
the end of the analysis.

\item Obtaining small upper bounds (e.g., $\exp(-c n)$ for some $c < 1$)
on the probability of various bad events in the analysis (e.g., $\dec(f(\enc(s)))$
significantly deviating from $\cD_f$) is not difficult to achieve.
However, extra care is needed to ensure that the probabilities
are much smaller than $1/(2^k |\cF|)$ (to accommodate the final union bound),
where the latter may easily be doubly-exponentially small in $n$.
An exponential upper bound of $\exp(-c n)$ does not even suffice for
moderately large families of adversaries such as bit-tampering adversaries,
for which we have $|\cF| = 4^n$.
\end{enumerate}

\subsection*{$\bullet\ $ Complete proof of Theorem~\ref{thm:upperBound}}

%
%

First, observe that by construction, the minimum distance of the final
code is always greater than $\delta n$; that is, 
whenever $\dec(w_1) \neq \perp$ and $\dec(w_2) \neq \perp$ for any
pair of vectors $w_1 \neq w_2$, we have
\[
\distH(w_1, w_2) > \delta n,
\]
where $\distH(\cdot)$ denotes the Hamming distance. This is because whenever a codeword
is picked, its $\delta n$ neighborhood is removed from the sample space for the future
codewords.
Let $V$ denote the volume of a Hamming ball of radius $\delta n$. It is
well known that $V \leq 2^{nh(\delta)}$, where $h(\cdot)$ is the binary entropy function.

Fix an adversary $f \in \cF$.
We wish to show that the coding scheme $(\enc, \dec)$ defined by Construction~\ref{constr:prob}
is non-malleable with high probability for
the chosen $f$. 

Define $p_0 := \Pr[f(U_n) = U_n]$.
In the sequel, assume that $p_0 < 1$ (otherwise, there
is nothing to prove). 
For every $x \in \zo^n$, define $p(x) := \Pr[f(U_n) = x \land x \neq U_n]$. 
Observe that
\[
\sum_x p(x) = 1 - p_0.
\]
We say that a string $x \in \zo^n$ is \emph{heavy} if 
\[
p(x) > 1/r,
\]
for a parameter $r \leq t$ to be determined later.  Note that
the number of heavy strings must be less than $r$. Define
\begin{align*}
H &:= \{ x \in \zo^n\colon p(x) > 1/r \}, \\
\gamma &:= t/N, \\
\gamma' &:= tK/N.
\end{align*}
Fix the randomness of the code construction so that $\dec(x)$
is revealed for every heavy $x$. 
We will argue that no matter how
the decoder's outcome on heavy elements is decided by the randomness of the code
construction, the construction is non-malleable for every message $s$
and the chosen function $f$ with overwhelming probability.
We will then finish the proof with a union bound over all choices of $s$ and $f$.

Consider a
random variable $D$ defined over $\zo^k \cup \{\perp, \same\}$ in the following way:
\begin{equation} \label{eqn:Df}
D := \begin{cases}
\same & \text{if $f(U_n) = U_n$}, \\
\dec(f(U_n)) & \text{if $f(U_n) \neq U_n$ and $f(U_n) \in H$}, \\
\perp & \text{otherwise.}
\end{cases}
\end{equation}
For the chosen $f$, we explicitly define the distribution $\cD_f$ 
as $\cD_f := \distr(D)$.

Now, consider a fixed message $s \in \zo^k$, and
define the random variable $E_s := \enc(s)$. 
That is, $E_s$ is uniformly supported 
on the set $E(s)$ (this holds by the way that the encoder is defined).
Observe that the marginal distribution of each individual set $E(s)$
(with respect to the randomness of the code construction) is the same
for all choices of $s$, regardless of the ordering assumed by
Construction~\ref{constr:prob} on the message space $\zo^k$.

Furthermore, define the random variable $D_s$ as follows.
\begin{equation}
\label{eqn:Ds}
D_s := \begin{cases}
\same & \text{if $f(E_s) = E_s$}, \\
\dec(f(E_s)) & \text{otherwise.}
\end{cases}
\end{equation}

Our goal is to show that the distribution of $D_s$ (for the final realization
of the code) is $\eps$-close to $\cD_f$
with high probability over the randomness of the code construction.
Such assertion is quite intuitive by comparing the way the two distributions $D_s$ and
$\cD_f$ are defined. In fact, it is not hard to show that the assertion holds
with probability $1-\exp(-\Omega(n))$. However, such a bound would be insufficient to
accommodate a union bound of even moderate sizes such as $2^n$, which is needed
for relatively simple classes such as bit-tampering adversaries. More work needs
to be done to ensure that it is possible to achieve a high probability statement
with failure probability much smaller than $1/|\cF|$, which may in general 
be doubly exponentially small in $n$.

The claim below shows that closeness of $\distr(D_s)$ to $\cD_f$ would imply
non-malleability of the code.

\begin{claim}
Suppose that for every $s \in \zo^k$, we have $\distr(D_s) \approx_\eps \cD_f$ for
the choice of $\cD_f$ defined in \eqref{eqn:Df}. Then, $(\enc, \dec)$ is a non-malleable
coding scheme with error $\eps$ and a strong non-malleable coding scheme with error $2\eps$.
\end{claim}

\begin{proof}
In order to verify Definition~\ref{def:nmCode:strong}, we need to verify that
for every distinct pair of messages $s_1, s_2 \in \zo^k$, $\distr(D_{s_1}) \approx_{2\eps}
\distr(D_{s_2})$. But from the assumption, we know that $\distr(D_{s_1})$ and 
$\distr(D_{s_2})$ are both $\eps$-close to $\cD_f$. Thus the result follows by the
triangle inequality.

It is of course possible now to use \cite[Theorem~3.1]{ref:nmc} to deduce that
Definition~\ref{def:nmCode} is also satisfied. However, for the clarity of presentation,
here we give a direct argument
that shows that non-malleability is satisfied with the precise choice of $\cD_f$
defined in $\eqref{eqn:Df}$ and error $\eps$.
Let $s \in \zo^k$, and let $E_s := \enc(s)$ and $S := \dec(f(E_s))$. 
Let $S' \sim \cD_f$ and $S'' \sim \distr(D_s)$ be sampled independently. We need to show that
\begin{equation}
\label{eqn:upper:nmc:verify}
\distr(S) \approx_{\eps} \distr(\Copy(S', s)).
\end{equation}
From the definition of $D_s$ in \eqref{eqn:Ds}, since $\dec(f(E_s)) = s$
when $f(E_s) = E_s$, we see that $\distr(\Copy(S'', s)) = \distr(\dec(f(E_s))) = \distr(S)$.
Now, since by assumption $\distr(S') \approx_\eps \distr(S'')$, it follows that
$\distr(\Copy(S', s)) \approx_\eps \distr(\Copy(S'', s))$ which proves \eqref{eqn:upper:nmc:verify}.
\end{proof}

Let the random variables $E_1, \ldots, E_t$ be the elements of $E(s)$, in the order they are sampled
by Construction~\ref{constr:prob}.

Define, for $i \in [t]$,
\[
S_i := \begin{cases}
\same & \text{if $f(E_i) = E_i$}, \\
\dec(f(E_i)) & \text{otherwise.} \\
\end{cases}
\]
We note that, no matter how the final code is realized by the randomness of the
construction, the distribution $D_s$ is precisely the empirical 
distribution of $S_1, \ldots, S_t$ as determined by the code construction.

In the sequel, for each $i \in [t]$, we analyze the distribution of
the variable $S_i$ conditioned on the values of $S_1, \ldots S_{i-1}$ and
use this analysis to prove that the empirical distribution of the
sequence $(S_1, \ldots, S_t)$ is close to $\cD_f$.

In order to understand the empirical distribution of the $S_i$, we consider
the following process $\reveal$ that considers the picked codewords
$E_1, \ldots, E_t$ in order, gradually reveals information about the code construction,
and outputs a subset of the $S_i$. We will ensure that
\begin{enumerate}
\item The process outputs a large subset of $\{S_1, \ldots, S_t\}$, and, 
\item The empirical distribution of the sequence output by the
process is close to $\cD_f$ with high probability.
\end{enumerate}
The above guarantees would in turn imply that the empirical distribution of
the entire sequence $S_i$ is also close to $\cD_f$ with high probability.
We define the process as follows.

\newcommand{\skp}{\mathsf{Skip}}
\newcommand{\lineSkpA}{\mbox{2}}
\newcommand{\lineSkpB}{\mbox{3.2.2}}
\newcommand{\lineUnveil}{\mbox{3.3.3}}

\subsubsection*{Process $\reveal$:}
\begin{description}
\item[1.] Initialize the set $\skp \subseteq [t]$ with the empty set.
Recall that the values of $\dec(w)$ for all $w \in H$ are already
revealed in the analysis, as well as $\dec(\Gamma(w))$ for those
for which $\dec(w) \neq \perp$.

\item[\lineSkpA.] For each heavy element $w \in H$, if $\dec(w) = s$,
consider the unique $j \in [t]$ such that $E_j = w$. Reveal\footnote{%
In a rigorous sense, by revealing a random variable we mean that we condition the
probability space on the event that a particular value is assumed by
the variable. For example, revealing $E_i$ means that the analysis branches to
a conditional world where the value of $E_i$ is fixed to the revealed value.
In an intuitive way, one may think of a reveal as writing constraints on the
realization of the code construction on a blackboard, which is subsequently
consulted by the analysis (in form of the random variable $\reveal_i$ that
the analysis defines to denote the information revealed by the process before stage $i$).
} $j$
and $E_j$, and add $j$ to $\skp$. 

\item[3.] For $i$ from $1$ to $t$, define \emph{the $i$th stage} as follows:

\begin{description}
\item[3.1.] If $i \in \skp$, declare a \emph{skip} and continue the loop with the next $i$. Otherwise,
follow the remaining steps.

\item[3.2.] Reveal $\Gamma(E_i)$. Note that revealing $E_i$ implies that $\dec(E_i)$ is revealed
as well, since $\dec(E_i) = s$. Moreover, recall that for any $x \in \Gamma(E_i) \setminus E_i$,
$\dec(x) = \perp$ by the code construction.

\item[3.3.] If $\dec(f(E_i))$ is not already revealed:
\begin{description}
\item[3.3.1.] Reveal $\dec(f(E_i))$.

\item[\lineSkpB.] If $\dec(f(E_i)) = s$, consider the unique $j \in [t]$ such that
$E_j = f(E_i)$. It must be that $j > i$, since $\dec(E_j)$ has not been revealed before.
Reveal $j$ and add it to $\skp$.

\item[\lineUnveil.] Declare that an \emph{unveil} has happened if
$\dec(f(E_i)) \neq \perp$. If so, reveal $\dec(f(x))$ for all $x \in \Gamma(f(E_i)) \setminus E_i$
to equal $\perp$.
\end{description}
\item[3.4.] Reveal and output $S_i$.
\end{description}
\end{description}

\newcommand{\nxt}{\mathsf{Next}}
For $i \in [t]$, we use the notation $\reveal_i$ to refer to all the information revealed
from the beginning of the process up to the time the $i$th stage begins.
We also denote by $\nxt(i)$ the least $j > i$ such that a skip does not occur
at stage $j$; define $\nxt(i) := t+1$ if no such $j$ exists, and
define $\nxt(0)$ to be the index of the first stage that is not skipped.
Moreover, for $w \in \zo^n$, we use the notation $w \in \reveal_i$ as
a shorthand to denote the event that the process $\reveal$ has revealed the
value of $\dec(w)$ at the time the $i$th stage begins.

By the way the code is constructed, the decoder's value at each given point
is most likely $\perp$. We make this intuition more rigorous and show that
the same holds even conditioned on the information reveal by the process
$\reveal$.

\begin{claim} \label{claim:upper:bottom}
For all $i \in [t]$ and any $a \in \supp(\reveal_i)$,
\[
\Pr[\dec(x) \neq \perp | (\reveal_i = a) \land (x \notin \reveal_i)] \leq \gamma'/(1-3 \gamma V).
\]
\end{claim}

\begin{proof}
Suppose $x \notin \reveal_i$, and 
observe that $\reveal_i$ at each step reveals at most the values
of the decoder at $2V$ points; namely, $\Gamma(E_i)$ and $\Gamma(f(E_i))$. 
Moreover, before the first stage, decoder's value is revealed at up to
$r$ heavy points and its Hamming neighborhood at radius $\delta n$.
In total, the total number of points at which decoder's value is revealed
by the information in $\reveal_i$ is at most
\[
(|H| + 2(i-1)) |V| \leq (2t+r)V \leq 3 \gamma V N.
\]
Let
\[\cC := \bigcup_s E(s) \]
be the set of all codewords of the coding scheme. Some of the elements of $\cC$
are already included in $\reveal_i$, and by assumption we know that none of
these is equal to $x$. 

The distribution of each unrevealed codeword, seen in isolation,
is uniform over the $N(1-3 \gamma V)$ remaining vectors in $\zo^n$.
Thus by taking a union bound on the probability of each such codeword
hitting the point $x$ (which is the only way to make $\dec(x) \neq \perp$, 
we deduce that
\[
\Pr[\dec(x) \neq \perp | \reveal_i = a] \leq \frac{tK}{N(1-3 \gamma V)} = \gamma'/(1-3 \gamma V). \qedhere
\]
\end{proof}

Ideally, for each $i \in [t]$ we desire to have $E_i$ almost uniformly distributed,
conditioned on the revealed information, so that the distribution of $\dec(f(E_i))$
(which is described by $S_i$ when $E_i$ does not hit a fixed point of $f$)
becomes close to $\dec(f(U_n))$. However, this is not necessarily true; for example,
when the process $\reveal$ determines the decoder's value on the heavy elements, 
the value of, say, $E_1$ may be revealed, at which point there is no hope to 
ensure that $E_1$ is nearly uniform. This is exactly what the set $\skp$ is designed for,
to isolate the instances when the value of $E_i$ is already determined by the 
prior information. More precisely, we have the following.

\begin{claim} \label{claim:upper:EiUniform}
Suppose that $i \notin \skp$ when the $i$th stage of $\reveal$ begins. Then,
for any $a \in \supp(\reveal_i)$,
\[
\distr(E_i | \reveal_i = a) \approx_{\nu} \U_n,
\]
where $\nu := (3\gamma V)/(1-3\gamma V)$.
\end{claim}

\begin{proof}
Note that, without any conditioning, the distribution of $E_i$ is exactly
uniform on $\zo^n$. If at any point prior to reaching the $i$th stage
it is revealed that $\dec(E_i) = s$, either line~\lineSkpA\ or line~\lineSkpB\ of
process $\reveal$ ensures that $i$ is added to the set $\skp$.

If, on the other hand, the fact that $\dec(E_i) = s$ has not been revealed
when the $i$th stage begins, the distribution of $E_i$ becomes uniform on
the points in $\zo^n$ that have not been revealed yet.
As in Claim~\ref{claim:upper:bottom}, the number of revealed points is
at most $(2t+r)V \leq 3 \gamma V N$. Thus, the conditional distribution $E_i$ remains
$((3\gamma V)/(1-3\gamma V))$-close to uniform by Proposition~\ref{prop:uniformity}.
\end{proof}

For each $i \in [t]$, define a random variable $S'_i \in \zo^k \cup \{\same, \perp\}$
as follows (where $U_n$ is independently sampled from $\cU_n$):
\begin{equation}
\label{eqn:Sprime}
S'_i := \begin{cases}
\same & \text{if $f(U_n) = U_n$}, \\
\dec(f(U_n)) & \text{if $f(U_n) \neq U_n \land f(U_n) \in \reveal_i$}, \\
\perp & \text{otherwise.}
\end{cases}
\end{equation}
Note that $\distr(S'_1) = \cD_f$.

Intuitively, $S'_i$ is the ``cleaned up'' version of the random variable $S_i$
\mchOK{Rephrase and comment on claim 7}
that we are interested in. As defined, $S'_i$ is an independent random variable, and
as such we are more interested in its \emph{distribution} than
value.
Observe that the distribution of $S'_i$ is randomly determined
according to the randomness of the code construction (in particular,
the knowledge of $\reveal_i$ completely determines $\distr(S'_i)$).
The variable $S'_i$ is defined so that its distribution approximates
the distribution of the actual $S_i$ conditioned on the revealed information
before stage $i$.
Formally, we can show that conditional distributions
of these two variables are (typically) similar. Namely, 

\begin{claim} \label{claim:upper:SvsSp}
Suppose that $i \notin \skp$ when the $i$th stage of $\reveal$ begins. Then,
for any $a \in \supp(\reveal_i)$,
\[
\distr(S_i | \reveal_i = a) \approx_\nu \distr(S'_i | \reveal_i = a),
\]
where $\nu := (3\gamma V+\gamma')/(1-3\gamma V)$.
\end{claim}

\begin{proof}
First, we apply Claim~\ref{claim:upper:EiUniform} to ensure that
\[
\distr(E_i | \reveal_i = a) \approx_{\nu'} \U_n,
\]
where $\nu' = (3\gamma V)/(1-3\gamma V)$.
Thus we can assume that the conditional distribution of $E_i$ is exactly
uniform at cost of a $\nu'$ increase in the final estimate.

Now, observe that, conditioned on the revealed information, the way $S_i$
is sampled at stage $i$ of $\reveal$ can be rewritten as follows:
\begin{enumerate}
\item Sample $E_i \sim \U_n$.

\item If $f(E_i) = E_i$, set $S_i \leftarrow \same$.

\item Otherwise, if $f(E_i) \in \reveal_i$, set $S_i$ 
to $\dec(f(E_i))$ as determined by the revealed information.

\item Otherwise, reveal $\dec(f(E_i))$ (according to its
conditional distribution on the knowledge of $\reveal_i$) and set $S$ accordingly.
\end{enumerate}

This procedure is exactly the same as how $S'_i$ is sampled by
\eqref{eqn:Sprime}; with the difference that at the third step, 
$S'_i$ is set to $\perp$ whereas $S_i$ is sampled according to the
conditional distribution of $\dec(f(E_i))$.
However, we know by Claim~\ref{claim:upper:bottom} that
in this case, 
\[
\Pr[\dec(f(E_i)) \neq \perp | \reveal_i = a] \leq \gamma'/(1-3 \gamma V).
\]
Thus we see that $S_i$ changes the probability mass of $\perp$
in $\distr(S'_i)$ by at most $\gamma'/(1-3 \gamma V)$. The claim
follows.
\end{proof}

Recall that the distribution of $S'_1$ is the same as $\cD_f$.
However, for subsequent stages this distribution may deviate from
$\cD_f$. We wish to ensure that by the end of process $\reveal$,
the deviation remains sufficiently small.

For $i \in [t-1]$, define $\Delta_i$ as 
\[
\Delta_i := \dist( \distr(S'_{i+1}), \distr(S'_i)).
\]
where $\dist(\cdot)$ denotes statistical distance.
Note that $\Delta_i$ is a random variable that is
determined by the knowledge of $\reveal_{i+1}$ 
(recall that $\reveal_i$ determines the exact distribution of $S'_i$).
We show that the conditional values attained by this random
variable are small in expectation.

\begin{claim} \label{claim:upper:Delta}
For each $i \in [t-1]$, and all $a \in \supp(\reveal_i)$,
\begin{equation} \label{eqn:upper:Deltai}
\E[\Delta_i | \reveal_i = a] \leq \frac{2 \gamma'}{r(1-3\gamma V)}.
\end{equation}
Moreover, $\Pr[\Delta_i \leq 2/r \mid \reveal_i = a] = 1$.
\end{claim}

\begin{proof}
Recall that the distribution of $S'_{i+1}$ is different from
$S'_i$ depending on the points at which the decoder's value
is revealed during stage $i$ of $\reveal$. If a skip is declared
at stage $i$, we have $\reveal_{i+1} = \reveal_i$ and thus,
$\Delta_i = 0$. Thus in the following we may assume that this is not the case.

However, observe that whenever for some $x \in \zo^n$, the decoder's
value $\dec(x)$ is revealed at stage $i$, the new information
affects the probability distribution of $S'_i$ only if $\dec(x) \neq \perp$.
This is because when $\dec(x) = \perp$, some of the probability mass
assigned by $S_i$ to $\perp$ in \eqref{eqn:Sprime} is removed and reassigned
by $S'_{i+1}$ to $\dec(x)$, which is still equal to $\perp$. Thus, changes of this type
can have no effect on the distribution of $S'_i$. We conclude that only
revealing the value of $\E_i$ and an unveil (as defined in line~\lineUnveil\ of
process $\reveal$) can contribute to
the statistical distance between $S'_i$ and $S'_{i+1}$.

Whenever an unveil occurs at stage $i$, say at point $x \in \zo^n$,
some of the probability mass assigned to $\perp$ by $S'_i$ is
moved to $\dec(x)$ in the distribution of $S'_{i+1}$. 
Since we know that $x \notin H$,
the resulting change in the distance between the two distributions
is bounded by $1/r$, \emph{no matter} what the realization of $x$ and
$\dec(x)$ are. Overall, using Claim~\ref{claim:upper:bottom},
the expected change between the two distributions contributed
by the occurrence of an unveil is upper bounded by the probability
of an unveil occurring times $1/r$, which is at most
\begin{equation}
\label{eqn:upper:perturb:a}
\frac{\gamma'/r}{1-3 \gamma V}.
\end{equation}

The only remaining factor that may contribute to an increase
in the distance between distribution of $S'_i$ and $S'_{i+1}$ 
is the revealing of $E_i$ at stage $i$.
The effect of this reveal in the statistical distance between
the two distributions is $p(E_i)$, since according to 
\eqref{eqn:Sprime} the value of $S'_{i+1}$ is determined by the
outcome of $f(U_n)$, and thus the probability mass assigned to
$\dec(E_i)$ by $S'_{i+1}$ is indeed $\Pr[f(U_n) = E_i]$.
Let $\cD_E$ be the distribution of $E_i$ conditioned on the
knowledge of $\reveal_i$. Observe that, since the values 
$\{ p(x)\colon x \in \zo^n \}$ 
defines a probability distribution
on $N$ points, we clearly have 
\begin{equation} \label{eqn:upper:pxForEi}
\sum_{x \in \supp(\cD_E)} p(x) \leq 1.
\end{equation}
On the other hand, by the assumption that 
a skip has not occurred at stage $i$,
we can deduce using the argument in Claim~\ref{claim:upper:EiUniform} that
$\cD_E$ is uniformly supported on a support of size at least $N(1-3\gamma V)$.
Therefore, using \eqref{eqn:upper:pxForEi}, 
the expected contribution to $\Delta_i$ by the
revealing of $E_i$ is (which is the expected value of $p(E_i)$) is 
at most 
\begin{equation}
\label{eqn:upper:perturb:b}
\frac{1}{N(1-3\gamma V)} \leq \frac{\gamma'/r}{(1-3\gamma V)},
\end{equation}
where the inequality uses $r \leq \gamma' N = tK$. The desired bound follows
by adding up the two perturbations 
\eqref{eqn:upper:perturb:a} and \eqref{eqn:upper:perturb:b}
considered.

Finally, observe that each of the perturbations considered above
cannot be more than $1/r$, since stage $i$ never
reveals the decoder's value on a heavy element (recall that
all heavy elements are revealed before the first stage begins
and the choices of $E_i$ that correspond to heavy elements
are added to $\skp$ when $\reveal$ begins). Thus, the conditional
value of $\Delta_i$ is never more than $2/r$.
\end{proof}

Using the above result, we can deduce a concentration bound on the
summation of the differences $\Delta_i$.

\begin{claim} \label{claim:upper:DeltaSum}
Let $\Delta := \Delta_1 + \cdots + \Delta_{t-1}$, and suppose
\begin{equation} \label{eqn:upper:assumption:a}
\frac{\gamma'}{1-3\gamma V} \leq \frac{\eps r}{32 t}.
\end{equation}
Then,
\begin{equation} \label{eqn:upper:eta0}
\Pr[ \Delta \geq \eps/8 ] \leq \exp(-\eps^2 r^2 /(2048 t)) =: \eta_0.
\end{equation}
\end{claim}

\begin{proof}
For $i \in [t-1]$, define $\Delta'_i := \Delta_i r/2$, $\Delta'_0 := 0$, 
and $\Delta' := \Delta'_1 + \cdots + \Delta'_{t-1}$.
Since $\reveal_i$ determines $\Delta'_{i-1}$, by Claim~\ref{claim:upper:Delta}
we know that
\[
\E[\Delta'_i | \Delta'_0, \ldots, \Delta'_{i-1} ] \leq \nu,
\]
where $\nu := \frac{\gamma'}{1-3\gamma V} \leq \eps r/(32 t)$
In the above, conditioning on $\Delta'_0, \ldots, \Delta'_{i-1}$ instead of
$\reveal_i$ (for which Claim~\ref{claim:upper:Delta} applies), is valid in light of Proposition~\ref{prop:restriction}, since the knowledge 
of $\reveal_i$ determines $\Delta'_0, \ldots, \Delta'_{i-1}$.

Moreover, again by the Claim~\ref{claim:upper:Delta}, we know that
the $\Delta'_i$ are between $0$ and $1$.
Using Proposition~\ref{prop:simpleAzuma}, it follows that
\[
\Pr[ \Delta \geq \eps/8 ] = \Pr[ \Delta' \geq \frac{\eps r}{16 t}\cdot t] 
\leq \eta_0. \qedhere
\]
\end{proof}

Next, we prove a concentration bound for the total number of unveils
that can occur in line~\lineUnveil\ of process $\reveal$.

\begin{claim} \label{claim:totalUnvel}
Let $u$ be the total number of unveils that occur in process $\reveal$.
Assuming
$
\gamma'/(1-3 \gamma V) \leq \eps/8
$ (which is implied by \eqref{eqn:upper:assumption:a}), we have
\[
\Pr[ u \geq \eps t/4 ] \leq \exp(-\eps^2 t/128) \leq \eta_0.
\]
\end{claim}

\begin{proof}
Let $X_1, \ldots, X_t$ be indicator random variable such that $X_i = 1$
iff an unveil occurs at stage $i$, and let $X_0 := 0$.
Recall that an unveil can only occur at a stage that is not skipped.
Thus, if $i \in [t]$ when the $i$th stage begins, we can deduce that
$X_i = 0$.

Consider $i \in [t]$ such that $i \notin \skp$ when the $i$th stage begins.
An unveil occurs when $\dec(f(E_i)) \notin \reveal_i$. In this case,
by Claim~\ref{claim:upper:bottom}, we get that
\[
\Pr[\dec(f(E_i)) \neq \perp | \reveal_i] \leq \gamma'/(1-3 \gamma V).
\]
Since $\reveal_i$ determines all the revealed information in each prior
stage, and in particular the values of $X_0, \ldots, X_{i-1}$, we can
use Proposition~\ref{prop:restriction} to deduce that
\[
\Pr[X_i = 1 | X_0, \ldots, X_{i-1}] \leq \gamma'/(1-3 \gamma V).
\]
Finally, Proposition~\ref{prop:simpleAzuma} derives the desired
concentration bound on the number of unveils, which is $X_1 + \cdots + X_t$.
\end{proof}

We are now ready to wrap up the proof and show that with overwhelming probability,
the empirical distribution of $S_1, \ldots, S_t$ is $\eps$-close to $\cD_f$.

Suppose that process $\reveal$ outputs a subset of the $S_i$. Let $T \subseteq [t]$
be the set of indices $i$ such $\reveal$ outputs $S_i$ in the end of the $i$th stage.
Note that $T = [t] \setminus \skp$, where $\skp$ denotes the skip set when $\reveal$ terminates.
Observe that $|\skp|$ is at most the total number of unveils occurring at line
\lineUnveil\ of $\reveal$ plus $r$ (which upper bounds the number of heavy elements in $H$). 
Thus, using Claim~\ref{claim:totalUnvel} we see that, assuming \eqref{eqn:upper:assumption:a},
\begin{equation}
\label{eqn:upper:skpSize}
\Pr[t-|T| \geq r+\eps t/4] \leq \eta_0.
\end{equation}
Let $\delta_i$ for $i \in [t]$ denote the statistical distance between
$S'_i$ and $\cD_f$. We know that $\delta_i$ is a random variable depending
on $\reveal_i$. Thus, the value of $\delta_i$ becomes known to a particular fixed
value conditioned on the outcome of every $\reveal_j$, $j \geq i$.
Define $\delta_0 := \max_i \delta_i$, which is a random variable that becomes
revealed by the knowledge of $\reveal_t$ in the end of the process.

Using Claim~\ref{claim:upper:SvsSp}, we thus know that for any $a \in \supp(\reveal_i)$
and $i \in T$,
\[
\distr(S_i | \reveal_i = a) \approx_{\nu_0+\delta_0} \cD_f,
\]
where \[\nu_0 := (3\gamma V+\gamma')/(1-3\gamma V).\]

Let $\cS$ denote the empirical distribution of $\{S_i\colon i \in T\}$,
and define $S_0 := \perp$.
From the above conclusion, using Proposition~\ref{prop:restriction} we can now write,
for $i \in T$,
\[
\distr(S_i | (S_j\colon j \in T \cap \{ 1, \ldots, i-1 \}) \approx_{\nu_0+\delta_0} \cD_f.
\]
Recall that $|\supp(\cD_f)| \leq r+2$. 
Assuming that
\begin{equation} \label{eqn:upper:assumption:b}
\nu_0 + \delta_0 \leq \eps/4,
\end{equation}
Proposition~\ref{lem:distrLearning:dependent} implies (after simple manipulations) that
with probability $1-\eta_1$, where
\begin{equation}
\label{eqn:upper:eta1}
\eta_1 \leq 2^{r+4-\Omega(\eps^2 |T|)},
\end{equation}
$\cS$ is $(\eps/2)$-close to $\cD_f$.

Recall that
$
\distr(S'_1) = \cD_f.
$
Using the triangle inequality for statistical distance, for every $i \in [t]$ we can write
\[
\dist(S'_i, \cD_f) = \dist(S'_i, S'_1) \leq \Delta_1 + \cdots + \Delta_{i-1} \leq \Delta,
\]
and thus deduce that
$
\delta_0 \leq \Delta.
$
Recall that by Claim~\ref{claim:upper:DeltaSum}, we can ensure that, assuming
\eqref{eqn:upper:assumption:a},
$\Delta \leq \eps/8$ (and thus, $\delta_0 \leq \eps/8$) with probability at least $1-\eta_0$. 
Thus under the assumption that
\begin{equation} \label{eqn:upper:assumption:c}
\nu_0 \leq \eps/8,
\end{equation}
and \eqref{eqn:upper:assumption:a}, which we recall below
\[
\frac{\gamma'}{1-3\gamma V} \leq \frac{\eps r}{32 t},
\]
we can ensure that $\nu_0 + \delta_0 \leq \eps/4$ with probability at least $1-\eta_0$.
Moreover, conditioned on the event $\nu_0 + \delta_0 \leq \eps/4$ (recall that $\delta_0$ 
is a random variable), we have already demonstrated that
with probability at least $1-\eta_1$, $\cS$ is $(\eps/2)$-close to $\cD_f$.
After removing conditioning on the bound on $\delta_0$, we may deduce that overall (under the assumed
inequalities \eqref{eqn:upper:assumption:a} and \eqref{eqn:upper:assumption:c}), with probability
at least $1-O(\eta_0 + \eta_1)$, 
\[
\cS \approx_{\eps/2} \cD_f,
\]
which in turn, implies that the empirical distribution of $S_1, \ldots, S_t$
becomes $\eps'$-close to uniform, where
\[
\eps' := \eps/2 + (1-|T|/t).
\]
Finally, we can use \eqref{eqn:upper:skpSize} to ensure that
(assuming \eqref{eqn:upper:assumption:a}), $\eps' \leq \eps$ and $|T|/t \geq 1-\eps/2$ with probability
at least $1-O(\eta_0+\eta_1)$ as long as
\begin{equation} \label{eqn:upper:assumption:d}
r \leq \eps t/4.
\end{equation}
By comparing \eqref{eqn:upper:eta1} with \eqref{eqn:upper:eta0}, we also deduce that
$\eta_1 = O(\eta_0)$ (and also that \eqref{eqn:upper:assumption:d} holds) as long as $r \leq r_0$ for some
\begin{equation} \label{eqn:upper:r0}
r_0 = \Omega(\eps^2 t).
\end{equation}

Altogether, we arrive at the conclusion that under assumptions 
\eqref{eqn:upper:assumption:a}, 
\eqref{eqn:upper:assumption:c}, and by taking
$r := r_0$, with probability at least $1-O(\eta_0)$, 
\[
(\text{empirical distribution of $(S_1, \ldots, S_t)$}) \approx_\eps \cD_f,
\]
which ensures the required non-malleability condition for message $s$
and tampering function $f$. By taking a union bound over all possible choices
of $s$ and $f$, the probability of failure becomes bounded by
\[
O(\eta_0 K |\cF|) =: \eta_2.
\]

We can now ensure that $\eta_2 \leq \eta$ for the chosen value for $r$ by taking
$t \geq t_0$ for some
\begin{equation} \label{eqn:T0}
t_0 = O\left( \frac{1}{\eps^6} \Big(\log\frac{|\cF|N}{\eta} \Big) \right).
\end{equation}

Furthermore, in order to satisfy assumptions \eqref{eqn:upper:assumption:a}, 
\eqref{eqn:upper:assumption:c}, and the requirement $tKV \leq 1$ which is
needed to make the construction possible,
it suffices to have $K \leq K_0$ for some
\[
K_0 = \Omega(\eps^3 N/(tV)).
\]
Using the bound $V \leq 2^{n h(\delta)}$, where $h(\cdot)$ is the binary
entropy function, and taking the logarithm of both sides, we see that
it suffices to have $k \leq k_0$ for some
\[
k_0 \geq n(1-h(\delta)) - \log t - 3 \log(1/\eps) - O(1).
\]

This concludes the proof of Theorem~\ref{thm:upperBound}.

\subsection{Efficiency in the random oracle model}
\label{sec:upper:efficiency}

One of the main motivations of the notion of non-malleable codes proposed
in \cite{ref:nmc} is the application for \emph{tamper-resilient security}.
In this application, a stateful consists of a public functionality and a
private state $s \in \zo^k$. The state is stored in form of its non-malleable
encoding, which is prone to tampering by a family of adversaries. It is shown
in \cite{ref:nmc} that the security of the system with encoded private state 
can be guaranteed (in a naturally defined sense) provided that the distribution
$\cD_f$ related to the non-malleable code is efficiently samplable. In light
of Remark~\ref{rem:efficiency}, efficient sampling of $\cD_f$ can be assured
if the non-malleable code is equipped with an efficient encoder and decoder.

Although the code described by Construction~\ref{constr:prob} may require
exponential time to even describe, it makes sense to consider efficiency
of the encoder and the decoder in the \emph{random oracle model}, where
all involved parties have oracle access to a shared, exponentially long,
random string. The uniform decoder construction of \cite{ref:nmc} is shown to be
efficiently implementable in the random oracle model in an 
\emph{approximate} sense (as long as all
involved parties query the random oracle a polynomial number of times),
assuming existence of an efficient algorithm implementing a
uniformly random permutation $\Pi$ and its inverse $\Pi^{-1}$.

We observe that Construction~\ref{constr:prob}, for the distance parameter
$\delta = 0$ (which is what needed for strong non-malleability as originally defined
in \cite{ref:nmc}) can be \emph{exactly} implemented efficiently 
(without any further assumptions on boundedness of the access to the random
oracle) assuming access to a uniformly random permutation and its
inverse (i.e., the so-called ideal-cipher model). 
This is because our code is designed so that the codewords are picked
uniformly at random and without replacement. More precisely,
the encoder, given message $s \in \zo^k$, can sample a uniformly
random $i \in [t]$, and output $\Pi(s, i)$, where $(s, i)$ is
interpreted as an element of $\zo^n$ (possibly after padding).

As noted in \cite{ref:nmc}, efficient approximate implementations of 
uniformly random permutations exist in the random oracle model. In particular,
\cite{ref:CPS08} show such an approximation with security $\poly(q)/2^n$,
where $q$ is the number of queries to the random oracle.

\section{A Monte Carlo construction for computationally bounded adversaries}
\label{sec:MC}

An important feature of Construction~\ref{constr:prob} is that the proof
of non-malleability, Theorem~\ref{thm:upperBound}, only uses limited independence
of the permutation defining the codewords $E(s)$ corresponding to each
message. This is because the proof analyzes the distribution of
$\dec(f(\enc(s)))$ for each individual message separately, and then
takes a union bound on all choice of $s$.

More formally, below we show that Theorem~\ref{thm:upperBound} holds for a 
broader range of code constructions than the exact Construction~\ref{constr:prob}.

\begin{defn}[$\ell$-wise independent schemes]
Let $(\enc, \dec)$ be any randomized construction of a coding scheme with block length
$n$ and message length $k$. For each $s \in \zo^k$, define $E(s) := \supp(\enc(s))$
and let $t_s := |\supp(\enc(s))|$.
We say that the construction is $\ell$-wise independent if the following are satisfied.
\begin{enumerate}
\itemsep=0ex
\item For any realization of $(\enc, \dec)$, the distribution of $\enc(s)$ (with respect 
to the internal randomness of $\enc$)
is uniform on $\supp(\enc(s))$.

\item The distribution of the codewords defined by the construction is $\ell$-wise independent.
Formally, we require the following.
Let $\cC := \bigcup_{s \in \{0,1\}^k} \supp(\enc(s))$. 
Suppose the construction can be described
by a deterministic function\footnote{As an example, in Construction~\ref{constr:prob},
all the values $t_s$ are equal to the chosen $t$, and moreover, one can take
$E(s, i, \mathcal{O}) = \Pi(s, i)$, where $\Pi\colon \zo^k \times [2^{n-k}] \to \zo^n$
is a uniformly random bijection defined by the randomness of $\mathcal{O}$.
}
$E\colon \zo^k \times \mathds{N} \times \mathds{N} \to \zo^n$ such that
for a bounded random oracle $\mathcal{O}$ over $\mathds{N}$ (describing the random bits used
by the construction), the sequence
\[
(E(s, i, \mathcal{O}))_{s \in \zo^k, i \in [t_s]}
\]
enumerates the set $\cC$. Moreover, for any set of $t$ indices $S = \{ (s_j, i_j)\colon j \in [\ell], 
s_j \in \zo^k, i_j \in [t_s]\}$,
we have 
\[ \distr(E(s_1, i_1, \mathcal{O}), \ldots, E(s_{\ell}, i_{\ell}, \mathcal{O})) = \distr(\Pi(1), \ldots, \Pi(\ell)) \]
for a uniformly random bijection $\Pi\colon [2^n] \to \zo^n$.
\end{enumerate}
\end{defn}
\begin{lem} \label{lem:MC:twise}
Let $(\enc, \dec)$ be any randomized construction of a coding scheme with block length
$n$ and message length $k$. For each $s \in \zo^k$, define $E(s) := \supp(\enc(s))$.
Suppose that for any realization of $(\enc, \dec)$, and for every $s_1, s_2 \in \zo^k$,
we have
\begin{enumerate}
\item $|E(s_1)| \geq t_0$, where $t_0$ is the parameter defined in Theorem~\ref{thm:upperBound}.

\item $|E(s_2)| = O(|E(s_1)|)$.
\end{enumerate}
Moreover, suppose that $k \leq k_0$, for $k_0$ as in Theorem~\ref{thm:upperBound}.
Let $t := \max_s |E(s)|$. 
Then, assuming that the construction is $(3t)$-wise independent, the conclusion of
Theorem~\ref{thm:upperBound} for distance parameter $\delta = 0$ 
holds for the coding scheme $(\enc, \dec)$.
\end{lem}

\begin{proof}
We argue that the proof of Theorem~\ref{thm:upperBound} holds without any
technical change if 
\begin{enumerate}
\item The codewords in $\supp(\enc(\cU_k))$ are chosen not fully independently but
$(3t)$-wise independently, and 

\item Each set $E(s)$ is not necessarily of exact size $t$ but of size
at least $t_0$ and $\Theta(t)$. 
\end{enumerate}
The key observation to be made is that the proof analyzes each individual
message $s \in \zo^k$ separately, and then applies a union bound on all choices
of $s$. Thus we only need sufficient independence to ensure that the
view of the analysis on each individual choice of the message is
statistically the same as the case where the codewords are chosen fully independently.

Observe that the bulk of the information about the code looked up by the analysis
for analyzing each individual message is contained in the random
variable $\reveal_{t+1}$ defined in the proof of Theorem~\ref{thm:upperBound}, 
that is defined according to how the process $\reveal$ evolves.
Namely, $\reveal_{t+1}$ summarizes all the information revealed about the
code by the end of the process $\reveal$.

For a fixed message $s \in \zo^n$ the process $\reveal$ iterates for
$|E(s)| \leq t$ step.
At each step, the location of at most two codewords in $\supp(\enc(\U_k))$
is revealed. Moreover, before the process starts, the values of the
decoder on the heavy elements in $H$, which can correspond to less than
$t$ codewords, are revealed by the process. The only other place
in the proof where an independent codeword is required is the union bound
in the proof of Claim~\ref{claim:upper:bottom}, which needs another
degree of independence. Altogether, we conclude that the proof of
Theorem~\ref{thm:upperBound} only uses at most $3t$ degrees of independence
in the distribution of the codewords picked by the construction. 

Moreover, for each message $s$, the analysis uses the fact that 
$|E(s)| \geq t_0$ to ensure that the code does not satisfy
non-malleability for the given choice of $s$ and tampering function
remains below the desired level. Since $|E(s)|$ for different values of $s$
are assumed to be within a constant factor of each other, the requirement
\eqref{eqn:upper:r0} may also be satisfied by an appropriate
choice of the hidden constant.
Finally, using the fact that $\max_s |E(s)| = O(\min_s |E(s)|)$,
we can also ensure that assumptions \eqref{eqn:upper:assumption:a}, and
\eqref{eqn:upper:assumption:c} can be satisfied for appropriate choices of
the hidden constants in asymptotic bounds.
\end{proof}

In order to implement an efficient $\ell$-wise independent coding scheme, we use
the bounded independence property of polynomial evaluations over finite fields. 
More precisely, we consider the coding scheme given in Construction~\ref{constr:monte}.

The advantage of using the derandomized Monte Carlo construction is that the number of
random bits required to describe the code is dramatically reduced from
$O(tnK)$ bits (which can be exponentially large if the rate of the code
is $\Omega(1)$) to only $O(tn)$ bits, which is only polynomially large
if $t = \poly(n)$. In order to efficiently implement the derandomized construction,
we use bounded independence properties of polynomial evaluation. Using known
algorithms for finite field operations and root finding, the implementation
can be done in polynomial time.

\newcommand{\encmc}{{\mathsf{EncMC}}}
\newcommand{\decmc}{{\mathsf{DecMC}}}
\begin{constr} 
\caption{The Monte Carlo Construction.}
  \begin{itemize}
  \item {\it Given: } Integer parameters $0 < k \leq n$ and integer $t > 1$
  which is a power of two. Let $b := \log(2t)$ and $m := n-k-b$.  

  \item {\it Output: } A coding scheme $(\encmc, \decmc)$ of block length $n$
  and message length $k$.

  \item {\it Randomness of the construction: } A uniformly random polynomial
  $P \in \F_{2^n}[9t-1]$.
  
  \item {\it Construction of $\encmc$: } 
  Given $s \in \zo^k$, 
  \begin{enumerate}
  \item Initialize a set $E \subseteq \zo^n$ to the empty set.
  
  \item For every $z \in \zo^b$,
  \begin{enumerate}
  \item Construct a vector $y := (s, 0^m, z) \in \zo^n$ and regard
  it as an element of $\F_{2^n}$. 
  
  \item Solve $P(X) = y$, and add the set
  of solutions (which is of size at most $9t-1$) to $E$.
  \end{enumerate}
  
  \item Output a uniformly random element of $E$.
\end{enumerate}   

  \item {\it Construction of $\decmc$: } 
  Given $x \in \zo^n$, interpret $x$ as an element of $\F_{2^n}$, and let
  $y := P(x)$, interpreted as a vector $(y_1, \ldots, y_n) \in \zo^n$. 
  If $(y_{k+1}, y_{k+2}, \ldots, y_{k+m}) = 0^m$, output $(y_1, \ldots, y_k)$.
  Otherwise, output $\perp$.
    
  \end{itemize}
\label{constr:monte}
\end{constr}

\begin{lem} \label{lem:MC}
Consider the pair $(\encmc, \decmc)$ defined in Construction~\ref{constr:monte}.
For every $\eta > 0$, there is a $t_0 = O(n+\log(1/\eta))$ such that for every $t \geq t_0$
(where $t$ is a power of two),
with probability at least $1-\eta$ the following hold.
\begin{enumerate}
\itemsep=0ex
\item $(\encmc, \decmc)$ is a $(9t)$-wise independent coding scheme.
\item For all $s \in \zo^k$, $|\supp(\encmc(s))| \in [t, 3t]$.
\end{enumerate}
\end{lem}

\begin{proof}
Let $N := 2^n$ and $K := 2^k$. 
Consider the vector $X := (X_1, \ldots, X_N) \in \F_{2^n}^N$, where $X_i := P(i)$ and
each $i$ is interpreted as an element of $\F_{2^n}$. Since the polynomial $P$ is of degree $9t-1$, 
the distribution of $X_1, \ldots, X_N$ over the randomness of the polynomial $P$ is
$(9t)$-wise independent with each individual $X_i$ being uniformly distributed on $\F_{2^n}$.
This standard linear-algebraic fact easily follows from 
invertibility of square Vandermonde matrices. 

Note that the decoder function $\decmc$ in Construction~\ref{constr:monte} is defined so that
\begin{equation} \label{eqn:mc:Puniform}
\decmc(U_n) = \begin{cases} 
\perp & \text{with probability $1-2tK/N$} \\
s \in \zo^k & \text{with probability $2t/N$}.
\end{cases}
\end{equation}

For $s \in \zo^k$, let $E(s) := \supp(\encmc(s))$. Note that the encoder, given $s$,
is designed to output a uniformly random element of $E(s)$.
Since the definition of the $\encmc(s)$ is so that it exhausts the list of all possible
words in $\zo^n$ that can lie in $\decmc^{-1}(s)$, it trivially follows that
$(\encmc, \decmc)$ is always a valid coding scheme; that is, for any realization of the
code and for all $s \in \zo^n$, we have $\decmc(\encmc(s)) = s$
subject to the guarantee that $|E(s)| > 0$.

Fix some $s \in \zo^k$.
Let $Z_1, \ldots, Z_N \in \zo$ be indicator random variable such that
$Z_i = 1$ iff $\decmc(i) = s$ (when $i$ is interpreted as an $n$-bit string). 
Recall that $(Z_1, \ldots, Z_N)$ is a
$(9t)$-wise independent random vector with respect to the randomness of the
code construction. Let $Z := Z_1 + \cdots + Z_N$, and note that
$Z = |E(s)|$.
From \eqref{eqn:mc:Puniform}, we see that
\[ \E[Z] = \E[|E(s)|] = 2t \ .\]
Using Theorem~\ref{thm:BR94} with $\ell := t/4$ and $A := \E[Z]/2 = t$, we see that
\[
\Pr[ |Z-2t| \geq t ] \leq 8 (3/4)^{t/4}.
\]
By taking a union bound over all choices of $s \in \zo^k$, we conclude that
with probability at least $1-\eta_0$, where we define
$\eta_0 := 8 N (3/4)^{t/4}$,
the realization of $(\encmc, \decmc)$ is so that
\[
(\forall s \in \zo^k)\colon |E(s)| \in [t, 3t].
\]
This bound suffices to show the desired conclusion.
\end{proof}

By combining the above tools with Theorem~\ref{thm:upperBound}, we can derive the
following result on the performance of Construction~\ref{constr:monte}.

\begin{thm} \label{thm:MC}
Let $\cF\colon \zo^n \to \zo^n$ be any family of tampering functions.
For any $\eps, \eta > 0$, 
with probability at least $1-\eta$, the pair $(\encmc, \decmc)$ in
Construction~\ref{constr:monte} can be set up so achieve a
non-malleable coding scheme with respect to $\cF$ 
and with error $\eps$. Moreover, the scheme satisfies the following.
\begin{enumerate}
\item The code achieves 
$k \geq n - \log\log (|\cF|/\eta) - \log n - 9 \log(1/\eps) - O(1)$.

\item The number of random bits needed to specify the code is $O\Big( (n+\log(|\cF|/\eta)) n/\eps^6 \Big)$.

\item The encoder and the decoder run in worst case time $\poly(\log(|\cF|/\eta)n/\eps)$.

\end{enumerate}
\end{thm}

\begin{proof}
Let $t_0$ and $k_0$ be the parameters promised by Theorem~\ref{thm:upperBound}.
We instantiate Construction~\ref{constr:monte} with parameter $t := t_0$ and $k := k_0$.
Observe that this choice of $t$ is large enough to allow Lemma~\ref{lem:MC}
to hold. Thus we can ensure that, with probability at least $1-\eta$,
$(\encmc, \decmc)$ is a $(9t)$-wise independent coding scheme where,
for every $s \in \zo^k$, $|E(s)| \in [t_0, 3t_0]$.
Thus we can now apply Lemma~\ref{lem:MC:twise} to conclude that with
probability at least $1-2\eta$, $(\encmc, \decmc)$ is a strong non-malleable
code with the desired parameters.

The number of random bits required to represent the code is
the bit length of the polynomial $P(X)$ in Construction~\ref{constr:monte},
which is $9tn$. Plugging in the value of $t$ from \eqref{eqn:T0} gives the desired estimate.

The running time of the decoder is dominated by evaluation of the polynomial
$P(X)$ at a given point. Since the underlying field is of characteristic two,
a representation of the field as well as basic field operations can be computed 
in deterministic polynomial time in the degree $n$ of the extension using
Shoup's algorithm \cite{ref:Shoup}.

The encoder is, however, slightly more complicated as it needs to iterate
through $O(t)$ steps, and at each iteration compute
all roots of a given degree $9t-1$ polynomial. Again, since characteristic of the 
underlying field is small, this task can be performed in deterministic polynomial time
in the degree $9t-1$ of the polynomial and the degree $n$ of the extension
(e.g., using \cite{ref:VS92}). After plugging in the bound on $t$ from \eqref{eqn:T0},
we obtain the desired bound on the running time.
\end{proof}

As a corollary, we observe that the rate of the Monte Carlo construction can be
made arbitrarily close to $1$ while keeping the bit-representation of the code
as well as the running time of the encoder and decoder at $\poly(n)$ provided
that $\eps = 1/\poly(n)$ and $|\cF| = 2^{\poly(n)}$. In particular, we see that
the Monte Carlo construction achieves strong non-malleability even with respect
to such powerful classes of adversaries as polynomial-sized Boolean circuits (with
$n$ outputs bits) and virtually any interesting computationally bounded model.

\begin{remark} \label{rem:inverseEps}
Since in this construction the error $\eps$ is only polynomially small, for
cryptographic applications such as tamper-resilient security it is important
to set up the code so as to ensure that $1/\eps$ is significantly larger
than the total number of tampering attempts made by the adversary.
\end{remark}

\begin{caveat*}
We point out that any explicit coding scheme for computationally bounded models
(such as polynomial-sized Boolean circuits) necessarily implies an explicit lower bound for the
respective computational model. This is because a function in the restricted
model cannot be powerful enough to compute the decoder function, as otherwise,
the following adversary would violate non-malleability: 
\begin{quote}
Consider fixed tuples $(s_1, x_1), (s_2, x_2) \in \zo^k \times \zo^n$,
where $s_1 \neq s_2$, $\dec(x_1) = s_1$ and $\dec(x_2) = s_2$.
Given a codeword $x \in \zo^n$, compute $s := \dec(x)$.
If $s = s_1$, output $x_2$. If $s = s_2$, output $x_1$. Otherwise,
output $x$.
\end{quote} 
\end{caveat*}

\begin{remark}(Alternative Monte Carlo construction)
In addition to Construction~\ref{constr:monte}, it is possible to
consider a related Monte Carlo construction when polynomial evaluation
is performed at the encoder and root finding is done by the encoder.
More precisely, the encoder, given $s \in \zo^k$, may sample
$i \in [t]$ uniformly at random, and output $P(s, i)$ where
$(s, i)$ is interpreted as an element of $\F_{2^n}$ (possibly
after padding). The drawback with this approach is that
the rate of the code would be limited by $1/2$, since for larger rates
there is a noticeable chance that the encoder maps different messages 
to the same codeword.
\end{remark}

\section{Impossibility bounds} \label{sec:lower}

In this section, we show that the bounds obtained by Theorem~\ref{thm:upperBound}
are essentially optimal.
In order to do so, we consider three families of adversaries.
Throughout the section, we use $k$ and $n$ for the message length and
block length of coding schemes and define $N := 2^n$ and $K := 2^k$.

\subsection{General adversaries} \label{sec:lower:general}

The first hope is to demonstrate that Theorem~\ref{thm:upperBound} is the best possible
for \emph{every} family of the tampering functions of a prescribed size. We rule out
this possibility and demonstrate a family $\cF$ of tampering functions achieving
$\log\log |\cF| \approx n$ for which there is a non-malleable code achieving rate
$1-\gamma$ for arbitrarily small $\gamma > 0$.

Let $S \subseteq \zo^n$ be any set of size at least $N^{1-\alpha}$ and
at most $N/2$.
Consider the family $\cF$ of functions satisfying the property that
\[
(\forall f \in \cF) (\forall x \in S)\colon f(x) = x.
\]
We can take the union of such families over all choices of $S$; however,
for our purposes it suffices to define $\cF$ with respect to a single
choice of $S$.
Observe that
\[
|\cF| N^{N-|S|} \geq N^{N/2},
\]
which implies
\[
\log\log |F| \geq n-1.
\]
However, there is a trivial coding scheme that is non-malleable with zero
error for all functions in $\cF$. Namely, the encoder $\enc$ is a deterministic
function that maps messages to distinct elements of $S$, whereas the decoder $\dec$
inverts the encoder and furthermore, maps any string outside $S$ to $\perp$.
In this construction, we see that
\[
(\forall f \in \cF) (\forall x \in \zo^k)\colon \dec(f(\enc(x))) = x,
\]
since $f$ necessarily fixes all the points in $S$ (in particular,
in Definition \ref{def:nmCode} one can take $\cD_f := \distr(\same)$).
Finally, observe that the rate of this coding scheme is at least
$1-\gamma$. In fact, this result holds for any $\gamma \geq 1/n$, 
implying that the rate of the code can be made $1-o(1)$.

\subsection{Random adversaries} \label{sec:lower:random}

The observation in Section~\ref{sec:lower:general} rules out the hope for 
a general lower bound that only depends on the size of the adversarial family.
However, in this section we show that for ``virtually all'' families
of tampering functions of a certain size, Theorem~\ref{thm:upperBound}
gives the best possible bound. More precisely, 
we construct a family $\cF$ of a designed size $M$
as follows: For each $i \in [M]$, sample a
uniformly random function $f_i\colon \zo^n \to \zo^n$ and add
$f_i$ to the family. Since some of the $f_i$ may turn out to be
the same (albeit with negligible probability), $|\cF|$ may in general
be lower than $M$ (which can only make a lower bound stronger).

We prove the following.

\begin{thm}
For any $\alpha > 0$, there is an $M_0$ satisfying
\[
\log \log M_0 \leq \alpha n + O(\log n)
\]
such that with probability $1-\exp(-n)$, a random family $\cF$ with designed
size $M \geq M_0$ satisfies the following: There is no coding scheme
achieving rate at least $1-\alpha$ and error $\eps < 1$ that is non-malleable with respect
to the tampering family $\cF$.
\end{thm}

\begin{proof}
We begin with the following simple probabilistic argument:

\begin{claim} \label{clm:allStrings}
Let $\cC \subseteq [q]^N$ be a multi-set of vectors each chosen
uniformly and independently at random. For any integer $\ell \in [N]$ and
parameter $\gamma > 0$, there is an $M_0 = O(\ell q^\ell \log (qN/\gamma))$
such that as long as $|\cC| \geq M_0$, the following holds with probability
at least $1-\gamma$: For every $S \subseteq [N]$ with $|S| \leq \ell$,
the set of vectors in $\cC$ restricted to the positions picked by $S$ is 
equal to $[q]^{|S|}$.
\end{claim}

\begin{proof}
Fix any choice of the set $S$ (where, without loss of generality,
$|S| = \ell$) and let $\cC_S$ be the set of vectors in
$\cC$ restricted to the positions in $S$. 
For any $w \in [q]^{|S}$, we have
\[
\Pr[w \notin \cC_S] = \Big( 1-\frac{1}{q^\ell} \Big)^{|\cC|} \leq \exp(-\Omega(|\cC|/q^\ell)).
\]
By taking a union bound on all the choices of $w$ and $S$, the probability
that $\cC$ does not satisfy the desired property can be seen to be at most
\[
(qN)^\ell \exp(-\Omega(|\cC|/q^\ell)),
\]
which can be made no more than $\gamma$ for some
\[
|\cC| = O\Big(q^\ell (\ell \log (qN) + \log(1/\gamma))\Big) \ . \qedhere
\]
\end{proof}
Let $\gamma > 0$ be a parameter to be determined later.
By Claim~\ref{clm:allStrings}, with probability at least $1-\gamma$
over the randomness of the family $\cF$, we can ensure that for all
sets $S \subseteq \zo^n$ of size at most $2 N^{\alpha}$, and
for all functions $f_S\colon S \to \zo^n$, there is a function $f \in \cF$
that agrees with $f_S$ on all points in $S$. This guarantee holds
if we take $\cF \geq M_0$ for some 
\[
M_0 = O\Big( N^{(4N^{\alpha})} (8 N^\alpha \log(N/\gamma)) \Big).
\]
Overestimating the above bound yields
\[
\log \log M_0 \leq \alpha n + \log\log(N/\gamma) + O(1)
\]
which is at most $\alpha n + O(\log n)$ for $\gamma = \exp(-n)$.
Assuming that the family $\cF$ attains the above-mentioned property,
we now proceed as follows. 

Consider any coding scheme $(\enc, \dec)$ with block length $n$ and
message length $k$ which is non-malleable for the family $\cF$
randomly constructed as above and achieving rate at least
$1-\alpha$ for some $\alpha > 0$ and any non-trivial error $\eps < 1$.
For any message $s \in \zo^k$, let
\[
E(s) := \supp(\enc(s)) \subseteq [N]
\]
and observe that $E(s) \cap E(s') = \emptyset$ for all $s \neq s'$.
Observe that
\[
\E[ |E(U_k)| ] \leq N^{\alpha}
\]
by the disjointness property of the $E(s)$ and the assumption
on the rate of the code. By Markov's bound,
\[
\Pr[ |E(U_k)| \geq 2 N^{\alpha} ] < 1/2
\]
implying that for at least half of the choices of $s \in \zo^k$, we can assume
$|E(s)| \leq 4 N^{\alpha}$. Take two distinct vectors
$s_1, s_2 \in \zo^k$ satisfying this bound.

Now, let $S := E(s_1) \cup E(s_2)$, where $|S| \leq 2 N^\alpha$ as above.
Consider any $c_1 \in E(s_1)$ and $c_2 \in E(s_2)$ and define
$f_S\colon S \to \zo^n$ such that
\[
(\forall x \in E(s_1))\colon\ f_S(x) = c_2 \quad \text{and} \quad 
(\forall x \in E(s_2))\colon\ f_S(x) = c_1.
\]
By the choice of $\cF$, we know that there is $f \in \cF$ that
agrees with $f_S$ on all the points in $S$. This choice of the adversary
ensures that 
\[
\Pr[ \dec(f(\enc(s_1))) = s_2 ] = 1
\quad \text{and} \quad 
\Pr[ \dec(f(\enc(s_2))) = s_1 ] = 1
\]
with respect to the randomness of the encoder. Since the
two distributions $\dec(f(\enc(s_1)))$ and $\dec(f(\enc(s_2)))$ are
maximally far from each other and moreover, the adversary $f$
always tampers codewords in $E(s_1)$ and $E(s_2)$ to a codeword
corresponding to a different message, we conclude that there is no choice of $\cD_f$
in Definition~\ref{def:nmCode} that ensures non-malleability with
any error less than~$1$.
\end{proof}

\subsection{General adversaries acting on a subset of positions}
\label{sec:lower:alpha}

An important family of adversaries is the one that is only restricted
by the subset of bits it acts upon. More precisely, let $T \subseteq [n]$
be a fixed set of size $\alpha n$, for a parameter $\alpha \in (0,1)$.
For $x \in \zo^n$, we use the notation $x_T \in \zo^{|T|}$ for the restriction of $x$
to the positions in $T$. Without loss of generality, assume that $T$
contains the first $|T|$ coordinate positions so that $x = (x_T, x_{\bar{T}})$,
where $\bar{T} := [n] \setminus T$.
We consider the family $\cF_T$ of all functions $f\colon \zo^n \to \zo^n$
such that 
\[
f(x) = (g(x_T), x_{\bar{T}})
\]
for some $g\colon \zo^{|T|} \to \zo^{|T|}$.
Observe that $|\cF_T| \leq N^{(\alpha N^{\alpha})}$ which implies
$\log \log |\cF_T| \leq \alpha n$.

We prove the following lower bound, which is a variation of the classical
Singleton bound for non-malleable codes. What makes this variation much
more challenging to prove is the fact that 1) non-malleable codes
allow a randomized encoder, and 2) non-malleability is a more relaxed
requirement than error detection, and hence the proof must rule out the
case where the decoder does not detect errors (i.e., outputs a wrong
message) while still satisfies non-malleability.

\begin{thm} \label{thm:lower}
Let $T \subseteq [n]$ be of size $\alpha n$ and consider the family 
$\cF_T$ of the tampering functions that only act on the coordinate 
positions in $T$ (as defined above). Then, there is a $\delta_0 = O((\log n)/n)$
such that the following holds.
Let $(\enc, \dec)$ be any coding scheme which is non-malleable for the
family $\cF_T$ and achieves rate $1-\alpha+\delta$,
for any $\delta \in [\delta_0, \alpha]$ and error $\eps$.
Then, $\eps \geq \delta/(16 \alpha)$. In particular, when $\alpha$ and
$\delta$ are absolute constants, $\eps = \Omega(1)$.
\end{thm}

Before proving the theorem, we state the following immediate corollary.

\begin{coro} \label{coro:lower}
Let $\cF$ be the family of split-state adversaries acting on $n$ bits. 
That is, each $f \in \cF$ interprets the input as a pair $(x_1, x_2)$
where $x_2 \in \zo^{\lfloor n/2 \rfloor}$ and $x_2 \in \zo^{\lceil n/2 \rceil}$,
and outputs $(f_1(x_1), f_2(x_2))$ for arbitrary tampering functions
$f_1$ and $f_2$ (acting on their respective input lengths).

Moreover, for a fixed constant $\delta \in (0,1)$, let $\cF_\delta$ be the class of tampering
functions where $f \in \cF_\delta$ iff every bit of $f(x)$ depends on 
at most $\lfloor \delta n \rfloor$ of ths bits of $x$.

Let $(\enc_1, \dec_1)$ (resp., $(\enc_\delta, \dec_\delta)$ be any coding scheme which is non-malleable for
the class $\cF$ (resp., $\cF_\delta$) achieving error at most $\eps$ and rate $R$ (resp., $R_\delta$).
Then, for every fixed constant $\gamma > 0$, there is a fixed constant
$\eps_0 > 0$ such that
if $\eps \leq \eps_0$, the following bounds hold.
\begin{enumerate}
\item[(i)] $R \leq 1/2 - \gamma$,

\item[(ii)] $R_\delta \leq 1 - \delta - \gamma$.
\end{enumerate}
\end{coro}

The proof of Theorem~\ref{thm:lower} uses basic tools from information
theory, and the core ideas can be described as follows.
Assume that the codeword is $(X_1, X_2)$ where the adversary
acts on $X_1$, which is of length $\alpha n$. We show that
for any coding scheme with rate slightly larger than $(1-\alpha)n$,
there is a set $X_\eta \subseteq \zo^{\alpha n}$ such that
\begin{enumerate}
\item For some message $s_0$, $X_1$ lies in $X_\eta$ with
noticeable probability.

\item For a ``typical'' message $s_1$, $X_1$ is unlikely to
land in $X_\eta$.

\item There is a vector $w \in \zo^{\alpha n}$ that cannot be
extended to a codeword $(w, w')$ that maps to either $s_0$
or $s_1$ by the decoder.
\end{enumerate}
We then use the above properties to design the following 
strategy that violates non-malleability of the code: Given
$(X_1, X_2)$, if $X_1 \in X_\eta$, the adversary tampers
the codeword to $(w, X_2)$, which decodes to a message
outside $\{ s_0, s_1 \}$. This ensures that $\dec(f(\enc(s_0)))$
has a noticeable chance of being tampered to an incorrect message.
Otherwise, the adversary leaves the codeword unchanged, 
ensuring that $\dec(f(\enc(s_1)))$ has little chance of 
being tampered at all. Thus there is no choice for a
distribution $\cD_f$ that sufficiently matches both
$\dec(f(\enc(s_0)))$ and $\dec(f(\enc(s_1)))$.

\subsubsection*{Proof of Theorem~\ref{thm:lower}}

Throughout the proof, we use standard information theoretic 
tools, such as the notation $H(X)$ for the Shannon entropy of a
discrete random variable $X$ and $I(X; Y)$ for the mutual information between
discrete random variables $X$ and $Y$. 
We will need the following standard information-theoretic fact.

\begin{claim} \label{clm:entropy}
Suppose $H(X) \leq r$ and let $p(x) := \Pr[X=x]$. For any $\eta > 0$,
and define 
\[ X_\eta := \{ x\in \supp(X)\colon p(x) > \frac{1}{2^{r/(1-\eta)}} \}. \]
Then, $\Pr[X \in X_\eta] \geq \eta$ and $|X| < 2^{r/(1-\eta)}$.
\end{claim}
\begin{proof}
The upper bound on $|X_\eta|$ is immediate from the definition of $X_\eta$.
Let $\bar{X_\eta} := \supp(X) \setminus X_\eta$. We need to show that
$\Pr[X \in \bar{X_\eta}] \leq 1-\eta$. If this is not the case, we can write
\begin{align*}
H(X) &\geq
\sum_{x \in \bar{X_\eta}} p(x) \log(1/p(x)) \\
&\geq \sum_{x \in \bar{X_\eta}} \frac{r p(x)}{1-\eta} \\
&= \Pr[x \in \bar{X_\eta}] r/(1-\eta) > r,
\end{align*}
a contradiction.
\end{proof}

Suppose there is a coding scheme $(\enc, \dec)$ that is non-malleable
for the family $\cF_T$ and achieving rate at least $1-\alpha+\delta$,
for an arbitrarily small parameter $\delta \in (0, \alpha]$.
Let $S \sim \cU_k$, $X := \enc(S)$ and suppose $X = (X_1, X_2)$ where
$X_1 := X_T$ and $X_2 := X_{\bar{T}}$.

For any $s \in \zo^k$, define $E(s) := \supp(\enc(s))$. Observe that
\[
\E_S |E(S)| \leq N/N^{1-\alpha+\delta} = N^{\alpha - \delta}
\]
By Markov's bound, for any $\gamma \in (0,1]$,
\begin{equation} \label{eqn:entropy:markov}
\Pr[|E(S)| > N^{\alpha - \delta}/\gamma] <\gamma.
\end{equation}

By the assumption on rate, $H(S) \geq n(1-\alpha+\delta)$. Also,
$H(X_2|S) \leq H(X_2) \leq n-|T| = n(1-\alpha)$. Thus,
\[
I(X_2; S) = H(S) - H(S|X_2)
\]
Using the chain rule for mutual information,
\begin{align}
I(X_1; S) &= I(X_1, X_2 ; S) - I(X_2 ; S | X_1) \nonumber \\
&= (H(S) - H(S|X_1, X_2)) - (H(X_2|X_1) - H(X_2| S, X_1)) \nonumber \\
&\geq H(S) - H(X_2|X_1) \label{eqn:entropy:rel:a} \\
&\geq H(S) - H(X_2) \label{eqn:entropy:rel:b} \\
&\geq (1-\alpha+\delta)n - (1-\alpha)n = \delta n, \label{eqn:entropy:rel:c}
\end{align}
where \eqref{eqn:entropy:rel:a} holds because $S = \dec(X_1, X_2)$ and thus 
$H(S|X_1, X_2) = 0$, in addition to non-negativity of entropy;
\eqref{eqn:entropy:rel:b} uses the fact that conditioning does not increase entropy;
and \eqref{eqn:entropy:rel:c} holds because of the assumption on the rate of the
code and the length of $X_2$. From this, we can deduce that
\begin{equation*} 
H(X_1|S) = H(X_1) - I(X_1; S) \leq H(X_1) - \delta n.
\end{equation*}
Note that the latter inequality in particular implies that $H(X_1) \geq \delta n$, 
and that $\supp(X_1) \geq 2^{\delta n}$.
By Markov's bound,
\begin{equation} \label{eqn:entropy:markov:b}
|\{ s \in \zo^k\colon H(X_1|S=s) > (H(X_1) - \delta n)(1+4\gamma) \}| < \frac{2^k}{1+4\gamma} 
\leq (1-2\gamma)2^k.
\end{equation}
By combining \eqref{eqn:entropy:markov} and \eqref{eqn:entropy:markov:b} using a union bound,
there is a choice of $s_0 \in \zo^k$ such that
\begin{equation*}
|E(s_0)| \leq N^{\alpha - \delta}/\gamma,\text{ and, } H(X_1|S=s_0) \leq (H(X_1) - \delta n)(1+4\gamma).
\end{equation*}
We can take $\gamma := \delta/(8\alpha)$ so that the above becomes
\begin{equation} \label{eqn:entropy:X1}
|E(s_0)| \leq 8 \alpha N^{\alpha - \delta}/\delta,\text{ and, } 
H(X_1|S=s_0) \leq H(X_1) - \delta n/2.
\end{equation}

For a parameter $\eta > 0$, to be determined later, we can now apply
Claim~\ref{clm:entropy} to the conditional distribution of
$X_1$ subject to $S = s_0$ and construct a set $X_\eta \subseteq \zo^{\alpha n}$
such that 
\begin{align}
&\Pr[X_1 \in X_\eta | S = s_0] \geq \eta, \label{eqn:entropy:s0} \\
&|X_\eta| \leq 2^{(H(X_1) - \delta n/2)/(1-\eta)}. \nonumber
\end{align}

Let $\eta' := \Pr[X_1 \in X_\eta]$, and let $h(\cdot)$ denote the binary entropy
function. Using a simple information-theoretic
rule that follows from the definition of Shannon entropy, we can write
\begin{align}
H(X_1) &= h(\eta') + \eta' H(X_1 | X_1 \in X_\eta) + (1-\eta') H(X_1 | X_1 \notin X_\eta) \nonumber \\
&\leq h(\eta') + \eta' \cdot \frac{H(X_1) - (\delta/2) n}{1-\eta} + (1-\eta') H(X_1 | X_1 \notin X_\eta)  \label{eqn:entropy:rel:d} \\
&\leq h(\eta') + \eta' \cdot \frac{H(X_1) - (\delta/2) n}{1-\eta} + (1-\eta') H(X_1), \label{eqn:entropy:rel:e}
\end{align}
where \eqref{eqn:entropy:rel:d} is due to the upper bound on the support size of $X_\eta$
and \eqref{eqn:entropy:rel:e} holds since conditioning does not increase entropy.
After simple manipulations, \eqref{eqn:entropy:rel:e} simplifies to
\begin{equation}
\eta' \leq \frac{2 h(\eta') (1-\eta)}{\delta n - 2\eta H(X_1)} \leq \frac{2 h(\eta')}{n(\delta - 2\eta \alpha)}.
\end{equation}
Now, we take $\eta := \delta/(4\alpha)$, so that the above inequalities, combined
with the estimate $h(\eta') = O(\eta' \log(1/\eta'))$ yields
\begin{equation*}
h(\eta')/\eta' \geq \delta n/4 \Rightarrow \log(1/\eta') = \Omega(\delta n)
\Rightarrow \eta' \leq \exp(-\Omega(\delta n)).
\end{equation*}
From the above inequality, straightforward calculations ensure that
\begin{equation} \label{eqn:entropy:etaP}
\eta' \leq \eta/4 = \delta/(16 \alpha),
\end{equation}
as long as $\delta \geq \delta_0 = O((\log n)/n)$.

From \eqref{eqn:entropy:etaP}, recalling that $\eta' = \Pr[X_1 \in X_\eta]$ and using Markov's bound, 
\[
| \{ s\colon \Pr[X_1 \in X_\eta | S = s] > \eta/2 \} |/2^{k} < 1/2.
\]
Combined with \eqref{eqn:entropy:markov} and a union bound, 
there is a fixed $s_1 \in \zo^k$ such that
\begin{equation} \label{eqn:entropy:s1}
|E(s_1)| \leq 8 \alpha N^{\alpha - \delta}/\delta,\text{ and, } 
\Pr[X_1 \in X_\eta | S = s_1] \leq \eta/2.
\end{equation}

Assuming the chosen lower bound for $\delta$, we can also ensure that, using
\eqref{eqn:entropy:X1}, that
$|E(s_0) \cup E(s_1)| < N^\alpha$. Thus, there is a fixed string $w \in \zo^{\alpha n}$
that cannot be extended to any codeword in $E(s_0)$ or in $E(s_1)$; i.e., 
\[
\Pr[X_1 = w | (S = s_0) \lor (S = s_1)] = 0,
\]
which in turn implies
\begin{equation} \label{eqn:entropy:w}
(\forall x_2 \in \zo^{n(1-\alpha)})\colon \dec(w, x_2) \notin \{s_0, s_1\}.
\end{equation}
Now, we consider the following tampering strategy 
$f\colon \zo^{|T|} \times \zo^{n-|T|} \to \zo^{|T|} \times \zo^{n-|T|}$ 
acting on the coordinate
positions in $T$:
\begin{itemize}
\item Given $(x_1, x_2) \in \zo^{|T|} \times \zo^{n-|T|}$, 
if $x_1 \in X_\eta$, output $(w, x_2)$.
\item Otherwise, output $(x_1, x_2)$.
\end{itemize}

Suppose the coding scheme $(\enc, \dec)$ satisfied Definition~\ref{def:nmCode}
for a particular distribution $\cD_f$ over $\zo^n \cup \{\same, \perp\}$ for the
tampering function $f$.

Since $f$ does not alter any string with the first component outside $X_\eta$,
\eqref{eqn:entropy:s1} implies that
\begin{equation} \label{eqn:entropy:fSame}
\Pr[f(X_1, X_2) = (X_1, X_2) | S = s_1] \geq 1-\eta/2.
\end{equation}
On the other hand, by \eqref{eqn:entropy:s0} and \eqref{eqn:entropy:w},
\begin{equation} \label{eqn:entropy:fTamper}
\Pr[\dec(f(X_1, X_2)) \notin \{s_0, s_1\} | S = s_0 ] \geq \eta.
\end{equation}
By \eqref{eqn:entropy:fTamper} and Definition~\ref{def:nmCode}, 
$\cD_f$ must be $\eps$-close to a distribution
$D_0$ that assigns at most $1-\eta$ of the probability mass to $\{ \same, s_0, s_1 \}$.
On the other hand, by \eqref{eqn:entropy:fSame}, $\cD_f$ must be $\eps$-close to a distribution $D_1$ that assigns at least $1-\eta/2$ of the probability mass to $\{ \same, s_1 \}$.
Thus, the statistical distance between $D_0$ and $D_1$ is at least $\eta/2$
(from the distinguisher corresponding to the event $\{ \same, s_1 \}$).
By triangle inequality, however, $D_0$ and $D_1$ are $(2\eps)$-close.
Therefore, $\eps \geq \eta/4$ and the result follows.

\bibliographystyle{abbrv}
\bibliography{\jobname} 

\appendix

\section{Rate $1/2$ barrier for the uniform coding scheme.}
\label{sec:barrier}

Dziembowski et al.~\cite{ref:nmc} consider the uniformly random 
coding scheme $(\enc, \dec)$ in which the decoder $\dec$ maps
any given input $x \in \zo^n$ to a uniform and independent random string in $\zo^k$.
Moreover, the encoder, given $s \in \zo^k$, outputs a uniformly random element of
$\dec^{-1}(s)$.
In this section, we argue that the uniform
coding scheme cannot achieve a rate better than $1/2$ even with respect to very
simple tampering functions.

Suppose that the scheme is indeed non-malleable with error upper bounded by a small constant
(say $1/4$), and consider any bijective tampering function 
$f\colon \zo^n \to \zo^n$. For example,
one may think of $f$ as the function that flips the first bit of the input.
For simplicity, we assume that the coding scheme achieves strong non-malleability (as
proved by Dziembowski et al.~\cite{ref:nmc}.  Since the chosen tampering function
does not have any fixed points (i.e., $f(x) \neq x$), Definition~\ref{def:nmCode:strong}
implies that there is a choice of $\cD_f$ that has no support on $\{\same\}$,
and we can restrict to such a distribution.  However, it can be shown that
the argument extends to the weaker definition of non-malleability as well.

Let $X := \enc(U_k)$ and observe that $\distr(X) = \cU_n$, which
in turn implies that $\distr(f(X)) = \cU_n$. Consider $S := \dec(f(X))$.
Note that $\distr(S)$ is a random variable depending on the randomness
of the code construction (namely, it is the empirical distribution of the
truth table of the decoder). With respect to this randomness, we have
\[
\E[\distr(S)] = \cU_k.
\]
Moreover, with overwhelming probability, the realization of
the code is so that
\[
\distr(S) \approx_{o(1)} \cU_k.
\]
Suppose this is the case and fix the randomness of the code construction accordingly.

Since for every $s \in \zo^k$, we know that $\distr(\dec(f(\enc(s))))$ is close
(in the sense described by Definition~\ref{def:nmCode}) to $\cD_f$, it 
follows that the convex combination 
\[
\sum_{s \in \zo^k} \Pr[\dec(U_n) = s]\cdot \distr(\dec(f(\enc(s))))
\]
is equally close to $\cD_f$. 
But, since $f(\enc(\dec(\U_n))) = f(\U_n) = \U_n$,
the above convex combination is exactly $\distr(\dec(\U_n)) = \distr(S)$, which
we know is close to $\U_k$. 

Thus it follows that for every $s \in \zo^k$, 
\[
(s, \dec(f(\enc(s)))) \approx_{o(1)} (s, \cU_k),
\]
and, for $U \sim \cU_k$,
\begin{equation}
\label{eqn:lower:uniform:entropy}
(U, \dec(f(\enc(U)))) \approx_{o(1)} \cU_{2k}.
\end{equation}
Since $(U, \dec(f(\enc(U))))$ is a function of $\enc(U)$, 
we get 
\[
H(U, \dec(f(\enc(U)))) \leq n.
\]
On the other hand, \ref{eqn:lower:uniform:entropy} implies that the above
entropy is close to $2k$. Thus, $k \leq (n/2) (1+o(1))$.

\section{Useful tools}

In many occasions in the paper, we deal with a chain of correlated random variables
$0 = X_0, X_1, \ldots, X_n$ where we wish to understand an event depending on $X_i$
conditioned on the knowledge of the previous variables. That is, we wish to understand
\[
\E[f(X_i) | X_0, \ldots, X_{i-1}].
\]
The following proposition shows that in order to understand the above quantity,
it suffices to have an estimate with respect to a more restricted event than
the knowledge of $X_0, \ldots, X_{i-1}$. Formally, we can state the following,
where $X$ stands for $X_i$ in the above example and $Y$ stands for $(X_0, \ldots, X_{i-1})$.

\begin{prop} \label{prop:restriction}
Let $X$ and $Y$ be possibly correlated random variables and let
$Z$ be a random variable such that the knowledge of $Z$ determines $Y$; 
that is, $Y = f(Z)$ for some function $f$. Suppose that
for every possible outcome of the random variable $Z$, namely, for
every $z \in \supp(Z)$, and for some real-valued function $g$, we have 
\begin{equation} \label{eqn:prop:restriction}
\E[g(X) | Z = z] \in I.
\end{equation}
for a particular interval $I$. Then, for every $y \in \supp(Y)$,
\[
\E[g(X) | Y = y] \in I.
\]
Similarly, suppose for some distribution $\cD$, and all $z \in \supp(Z)$,
\[
\distr( X | Z = z ) \approx_\eps \cD.
\]
Then, for all $y \in \supp(Y)$,
\[
\distr( X | Y = y ) \approx_\eps \cD.
\]
\end{prop}
\begin{proof}
Let $T = \{ z \in \supp(Z)\colon f(z) = y \}$, and let 
$p(z) := \Pr[Z = z | Y = y]$. Then,
\[
\E[g(X) | Y = y] = \sum_{z \in T} p(z) \E[g(X) | Z = z].
\] 
Since by \eqref{eqn:prop:restriction}, each $\E[g(X) | Z = z]$ lies in $I$ 
and $\sum_{z \in T} p(z) = 1$, we deduce that
\[
\E[g(X) | Y = y] \in I.
\]
Proof of the second part is similar, by observing that if a collection
of distributions is statistically close to a particular distribution $\cD$,
any convex combination of them is equally close to $\cD$ as well.
\end{proof}

\begin{prop} \label{prop:uniformity}
Let the random variable $X \in \zo^n$ be uniform on a set of size
at least $(1-\eps)2^n$. Then, $\cD(X)$ is $(\eps/(1-\eps))$-close to $\U_n$.
\end{prop}

We will use the following tail bounds on summation of possibly dependent 
random variables, which are direct consequences of Azuma's inequality.

\begin{prop} \label{prop:simpleAzuma}
Let $0 = X_0, X_1, \ldots, X_n$ be possibly correlated random variables in $[0,1]$
such that for every $i \in [n]$ and for some $\gamma \geq 0$, 
\[
\E[X_i | X_0, \ldots, X_{i-1}] \leq \gamma.
\]
Then, for every $c \geq 1$,
\[
\Pr[\sum_{i=1}^n X_i \geq cn\gamma] \leq \exp(- n\gamma^2 (c-1)^2/2),
\]
or equivalently, for every $\delta > \gamma$,
\[
\Pr[\sum_{i=1}^n X_i \geq n\delta] \leq \exp(- n(\delta-\gamma)^2/2).
\]
\end{prop}

\begin{proof}
The proof is a standard Martingale argument. For $i \in [n]$, define
\[
X'_i := X_i - \gamma,
\]
and
\[
S_i := \sum_{j=1}^i X'_i  = \sum_{j-1}^i X_i - i \gamma.
\]
By assumption, $S_i$ is a super-martingale, that is, assuming $S_0 := 0$,
\[
\E[S_{i+1} | S_0, \ldots, S_i] \leq S_i.
\]
Thus, by Azuma's inequality, for all $t \geq 0$,
\[
\Pr[S_n \geq t] \leq \exp(-t^2/(2n)).
\]
Substituting $t := (c-1) n \gamma$ proves the claim.
\end{proof}

In a similar fashion (using Azuma's inequality for sub-martingales
rather than super-martingales in the proof), we may obtain a tail bound
when we have a lower bound on conditional expectations.

\begin{prop} \label{prop:simpleAzuma:sub}
Let $0 = X_0, X_1, \ldots, X_n$ be possibly correlated random variables in $[0,1]$
such that for every $i \in [n]$ and for some $\gamma \geq 0$,
\[
\E[X_i | X_0, \ldots, X_{i-1}] \geq \gamma.
\]
Then, for every $\delta < \gamma$,
\[
\Pr[\sum_{i=1}^n X_i \leq n\delta] \leq \exp(- n(\delta-\gamma)^2/2). \qedhere \]
\end{prop}

The following tail bound is similar in flavor to the one given
by Proposition~\ref{prop:simpleAzuma}, but only applies to indicator
random variables. However, it can be better
when the individual expectations are low and the target deviation
from mean is very large.

\begin{prop} \label{prop:chernoff:dependent}
Let $0=X_0, X_1, \ldots, X_n \in \zo$ be indicator, possibly dependent, random variables such that
for every $i \in [n]$, \[ \E[X_i | X_1, \ldots, X_{i-1} ] \leq p,\]
for some $p \in [0, 1]$. Let $X := X_1 + \cdots + X_n$. Then, for every $c \geq 1$,
\[
\Pr[X > cnp] \leq (e/c)^{cnp}.
\]
\end{prop}

\begin{proof}
We closely follow the standard proof of Chernoff bounds for independent indicator random variables
(see, e.g., \cite{ref:MR}). Using Markov's bound on the exponential moment of $X$, we can
write, for a parameter $t > 0$ to be determined later,
\begin{align} \label{eqn:exp:moment}
\Pr[X > cnp] &\leq \frac{\E[\exp(tX)]}{\exp(tcnp)} = \frac{\E[\exp(tX_1) \cdots \exp(tX_n)]}{\exp(tcnp)}.
\end{align}
However, we can write down the expectation of product as the following chain of conditional
expectations
\begin{align*}
\E_{(X_1, \ldots, X_n)}[\exp(tX)] &= \E_{X_1}\Big[e^{tX_1} \E_{(X_2|X_1)} \big[e^{tX_2} \ldots \E_{(X_n|X_1,\ldots,X_{n-1})}
 e^{tX_n}] \ldots\big]\Big] \\
 &\leq (p \exp(t) + 1)^n.
\end{align*}
where the inequality uses the fact that the $X_i$ are Bernoulli random variables and
thus 
\[\E[\exp( tX_{i} ) | X_1, \ldots, X_{i-1}] \leq p\exp(t) + (1-p) \exp(0) \leq p \exp(t) + 1.\]
Using the inequality $(1+x)^n \leq \exp(nx)$ the above simplifies to
\begin{align*}
\E[\exp(tX)] \leq \exp(np\exp(t)),
\end{align*}
and thus, plugging the above result into \eqref{eqn:exp:moment},
\[
\Pr[X > cnp] \leq \frac{\exp(np\exp(t))}{\exp(tcnp)}.
\]
Choosing $t := \ln c$ yields the desired conclusion.
\end{proof}

For summation of $\ell$-wise independent random variables, we use the following tail
bound from \cite{ref:BR94}:
\begin{thm} \label{thm:BR94}
Let $\ell > 1$ be an even integer, and let $X_1, \ldots, X_n \in [0, 1]$ be
$t$-wise independent variables. Define $X := X_1 + \cdots + X_n$ and
$\mu := \E[X]$. Then,
\[
\Pr[|X-\mu| \geq A] \leq 8 \Big( \frac{\ell(\mu+\ell)}{A^2} \Big)^{\ell/2}.
\qedhere \]
\end{thm}

\subsection*{Approximating distributions by fuzzy correlated sampling}

In this section, we show that it is possible to sharply approximate a distribution
$\cD$ with finite support by sampling possibly correlated random variables
$X_1, \ldots, X_n$ where the distribution of each $X_i$ is close to $\cD$
conditioned on the previous outcomes, and computing the empirical distribution
of the drawn samples.

\begin{lem} \label{lem:distrLearning:dependent}
Let $\cD$ be a distribution over a finite set $\Sigma$ such that
$|\supp(\cD)| \leq r$. For any
$\eta, \eps, \gamma > 0$ such that $\gamma < \eps$, there is a choice of 
\[
n = O((r+2+\log(1/\eta))/(\eps-\gamma)^2)
\] such that the following holds. 
Suppose
$0 = X_0, X_1, \ldots, X_n \in \Sigma$ are possibly correlated random variables such that
for all $i \in [n]$ and all values $0 = x_0, x_1 \ldots, x_n \in \supp(\cD)$,
\[
\distr(X_i | X_0 = x_0, \ldots, X_{i-1} = x_{i-1}) \approx_\gamma \cD.
\]
Then, with probability at least $1-\eta$,
the empirical distribution of the outcomes 
$X_1, \ldots, X_n$ is $\eps$-close to $\cD$.
\end{lem}

\begin{proof}
First, we argue that without loss of generality, we can assume that $|\Sigma| \leq r+1$. This is because
if not, we can define a function
$f\colon \Sigma \to \supp(\cD) \cup \{ \star \}$ as follows:
\[
f(x) := \begin{cases} 
x & \text{if $x \in \supp(\cD)$} \\
\star & \text{otherwise.}
\end{cases}
\]
Observe that for any distribution $\cD'$ over $\Sigma$, 
$\dist(\cD', \cD) = \dist(f(\cD'), \cD)$, since the elements outside
$\supp(\cD)$ always contribute to the statistical distance
and we aggregate all such mass on a single extra point $\star$,
and by doing so do not affect the statistical distance.
Thus the empirical distribution of $(X_1, \ldots, X_n)$ is $\eps$-close
to $\cD$ iff the empirical distribution of $(f(X_1), \ldots, f(X_n))$ is.

Now suppose $|\Sigma| \leq r+1$.
Let $A \subseteq \Sigma$ be any non-empty event, and denote by
$\cD'$ the empirical distribution of the outcomes $X_1, \ldots, X_n$.
Let $p := \cD(A)$, and define indicator random variables
\[   
Y_i := \begin{cases} 0 & X_i \notin A, \\ 1 & X_i \in A. \end{cases}
\]
for $i \in [n]$ and $Y_0 := 0$.
Observe that
\[
\cD'(A) = \frac{\sum_{i=1}^n Y_i}{n},
\]
and, by the assumption on the closeness of conditional distributions of the $X_i$
to $\cD$,
\[
\E[Y_i | Y_0, \ldots, Y_{i-1}] \in [p - \gamma, p + \gamma].
\]
By Propositions \ref{prop:simpleAzuma}~and~\ref{prop:simpleAzuma:sub}, we 
can thus obtain a concentration bound
\begin{align*}
\Pr[| \cD'(A) - p | > \eps] \leq 2 \exp(-(\eps-\gamma)^2 n/2).
\end{align*}
Now we can apply a union bound on all possible choices of $A$ and conclude that
\[
\Pr[\lnot (\cD' \approx_\eps \cD)] \leq 2^{r+2} \exp(-(\eps-\gamma)^2 n/2),
\]
which can be ensured to be at most $\eta$ for some choice of 
\[
n = O((r+2+\log(1/\eta))/(\eps-\gamma)^2). \qedhere
\]
\end{proof}

\end{document}